\documentclass[12pt, a4paper]{amsart}
\usepackage[utf8]{inputenc}
\usepackage{graphicx}
\usepackage{dcolumn}
\usepackage[hidelinks]{hyperref}
\usepackage[english]{babel}
\usepackage[autostyle, english = american]{csquotes}
\MakeOuterQuote{"}
\usepackage{amsthm}
\usepackage{mathrsfs}
\usepackage{amssymb}
\usepackage{bm}
\usepackage{physics}

\usepackage{cochineal}
\usepackage[cochineal]{newtxmath}

\usepackage[backend=bibtex,style=numeric, sorting = nyt, maxcitenames=5,mincitenames=3,maxbibnames=10,giveninits=true, doi=false, isbn = false]{biblatex}

\newtheorem{theorem}{Theorem}[section]

\newtheorem{proposition}[theorem]{Proposition}
\newtheorem{lemma}[theorem]{Lemma}
\theoremstyle{definition}
\newtheorem{definition}[theorem]{Definition}
\newtheorem{remark}[theorem]{Remark}
\newtheorem{example}{Example}[section]

\DeclareMathOperator{\BRST}{\gamma_{BRST}}
\DeclareMathOperator{\bv}{\bigtriangleup}
\DeclareMathOperator{\supp}{\text{supp}}

\newcommand{\mBV}{s_k}
\newcommand{\Slav}{\mathscr S_k}
\newcommand{\mBRST}{\mathscr S}

\newcommand{\Ocal}{\mathcal{O}}
\newcommand{\Mcal}{\mathcal{M}}
\newcommand{\conf}{\mathscr E}
\newcommand{\RR}{\mathbb{R}}
\newcommand{\ph}{\varphi}
\newcommand{\g}{\mathscr g}
\newcommand{\C}{\mathscr C}
\newcommand{\An}{A}

\newcommand{\F}{\mathscr{F}}
\newcommand{\V}{\mathscr{V}}
\newcommand{\loc}{{\rm loc}}
\newcommand{\Floc}{\F_{\loc}}

\newcommand{\A}{{A}}

\begin{document}
\title{A Lorentzian renormalisation group equation for gauge theories
}
\author{Edoardo D'Angelo}
\address{Dipartimento di Matematica, Università di Genova, Italy}
\address{Istituto Nazionale di Fisica Nucleare -- Sezione di Genova, Italy}
\email{edoardo.dangelo@edu.unige.it}

\author{Kasia Rejzner}
\address{Department of Mathematics, University of York, UK}
\email{kasia.rejzner@york.ac.uk}

\begin{abstract}
In a recent paper, with Drago and Pinamonti we have introduced a Wetterich-type flow equation for scalar fields on Lorentzian manifolds, using the algebraic approach to perturbative QFT. The equation governs the flow of the effective average action, under changes of a mass parameter $k$. Here we introduce an analogous flow equation for gauge theories, with the aid of the Batalin-Vilkovisky (BV) formalism.  
We also show that the corresponding effective average action satisfies a Slavnov-Taylor identity in Zinn-Justin form. We interpret the equation as a cohomological constraint on the functional form of the effective average action, and we show that it is consistent with the flow.
\end{abstract}
\maketitle
\section{Introduction}
Gauge symmetry and renormalization are two pillars of modern physics. General Relativity and Yang-Mills theory are built upon the invariance of the laws of Nature under diffeomorphisms and under a non-Abelian Lie group, respectively. In quantum field theory (QFT) language, the gauge symmetry of the classical action gives raise to Becchi-Rouet-Stora-Tyutin (BRST) symmetry after gauge fixing is imposed to quantise a theory. The invariance of scattering amplitudes under BRST symmetry, captured by the local cohomology of the BRST operator, is the fundamental requirement to obtain gauge-invariant, physical observables from a gauge-fixed action. Later, this was generalised to the Batalin-Vilkovisky (BV) formalism, which more recently has also been made rigorous in mathematical QFT \cite{CoGw-1, CoGw-2, Costello2007,Fredenhagen2013,Hollands2008}.

On the other hand, renormalization explains how the macroscopic world emerges from our microscopic models of fundamental interactions.  

In the Wilsonian view of renormalization \cite{Wilson73}, high energy degrees of freedom are progressively integrated out in the path-integral, obtaining a coarse-grained, finite description of a system from its microscopic properties. In one of its modern formulations, the functional Renormalization Group (fRG) \cite{Berges2000, Dupuis2021, Gies2006, Rosten2010}, coarse-graining is achieved by introducing a momentum-dependent cut-off in the action, which acts as a mass-like term for modes with momentum below an arbitrary scale $k$, and vanishes for higher-energy modes. The main object of study in the fRG is the effective average action, defined as the Legendre transform of the generating functional of connected Green's functions. In the limits $k \to \infty$ and $k \to 0$ the effective average action interpolates between the classical action and the full quantum action.  The flow equation then governs the $k-$dependence of the effective average action. The Wetterich equation \cite{Wetterich1992}, which is a widely used flow equation, is both UV and IR finite and can be taken as the starting point for the quantisation of a theory, since the functional derivatives of the effective average action give the n-points 1PI correlation functions.

The fRG has found numerous applications, from statistical physics to high-energy physics and quantum gravity (see e.g. \cite{Berges2000, Dupuis2021, Niedermaier2006, Percacci2017, ReuterSaueressig2012, ReuterSaueressig2019, Saueressig2023} for reviews). In particular, since its first developments the Wetterich equation has been extensively used to study gauge theories \cite{Ellwanger1994, Reuter1993, Litim1998, Pawlowski2005}. Here, the use of a momentum-dependent regulator explicitly breaks gauge invariance, manifesting itself as an additional, regulator-dependent term in the Ward identities. The arising modified identities, called modified Slavnov-Taylor identities (mSTI) in the fRG literature, are then used as a symmetry constraint along the flow, restricting the functional form of the effective average action. In most applications, however, the mSTI, just as the Wetterich equation, cannot be solved exactly. The breaking of the mSTI, then, provides a measure for the quality of the chosen approximation. 

Despite its successes, the fRG still lacks a fully mathematically rigorous formulation, although there are some recent results \cite{Ziebell2021} in that direction. However, it is possible in some cases to implement a mathematically rigorous Wilsonian flow. 
In the context of Constructive Renormalization Group
(CRG), we can refer to \cite{Giuliani2020} for a recent review and a worked fermionic theory example. 

In particular, the Constructive Renormalization Group approach has given a rigorous construction of non-trivial fixed points in different cases (see for example \cite{Abdesselam2006, Brydges1998, Brydges2002}, and the introductions \cite{Abdesselam2013, Bauerschmidt2019, Salmhofer1999, Giuliani2020}). All these examples show that a weakly coupled RG can be non-perturbatively implemented, and a fixed point can be found without any approximation in a Banach space of interactions. 

Both the fRG and the CRG are usually formulated in Euclidean spaces.
In a recent paper \cite{DDPR2022}, together with Pinamonti and Drago, we proposed a new flow equation for scalar fields, analogous to the Wetterich equation, which works in Lorentzian signature for a large class of backgrounds (i.e. on globally hyperbolic manifolds, and it has been adapted to manifolds with timelike boundary \cite{Dappiaggi2024}). This new flow equation can be studied using the standard fRG methods and its formulation involves applications of perturbative Algebraic Quantum Field Theory (pAQFT). The latter is a mathematically rigorous framework for perturbative QFT, which combines the Haag's idea of locality \cite{Haag1992} with perturbative methods used by practitioners of QFT.

In pAQFT, a quantum field theory model on a globally hyperbolic spacetime $\mathcal M$ is given as a net of formal power series (in $\hbar$ and potentially also the coupling constant $\lambda$) with coefficients in topological $*-$algebras assigned to relatively compact regions $\Ocal\subset \mathcal M$. The net satisfies the axiom of \textit{causality} ($[\mathscr A(\mathcal O_1) , \mathscr A(\mathcal O_2)] = 0$ if $\mathcal O_2$ is not in the causal future of $\mathcal O_1$) and typically \textit{isotony} (i.e., the algebra $\mathscr A(\mathcal O_1)$ associated to a subregion $\mathcal O_1 \subset \mathcal O$ of a larger region $\mathcal O$ is contained in the algebra of the larger region, $\mathscr A_1(\mathcal O_1) \subset \mathscr A(\mathcal O)$) and the \textit{time-slice axiom} (a quantum version of the well-posedness of the Cauchy problem). To be able to formulate a model on all globally hyperbolic spacetimes in a coherent way, one uses the framework of locally covariant quantum field theory \cite{Brunetti2001}.

For the treatment of gauge theories, the axioms stated above have to be slightly weakened, as algebras are replaced by differential graded algebras \cite{GwilliamRejzner2017, Benini2019, Benini2018}. Primary examples are provided by Yang-Mills theory \cite{Hollands2008,FR} and effective quantum gravity  \cite{Brunetti2016}.

In the algebraic approach, states are normalised, positive, linear functionals on the algebra of observables $\omega : \mathscr A \to \mathbb C$ and the Hilbert space representation is obtained via the GNS construction. Hence pAQFT is well-suited to handle generic states on possibly curved Lorentzian spacetimes, and as such it appears as a good candidate to overcome some of the difficulties of the standard fRG approach, in which the Wetterich equation is formulated only in the Euclidean setting and is limited to some distinguished classes of states, e.g. the vacuum or thermal states \cite{Litim2006}. In \cite{DDPR2022}, we showed that the flow equation for the scalar field formulated using pAQFT results in $\beta-$functions which are qualitatively in agreement with the known literature.

In this paper, we take a step further in this direction, developing a flow equation for gauge theories, using the Batalin-Vilkovisky formalism.

The Batalin-Vilkovisky (BV) formalism \cite{Batalin1977, Batalin1981, Batalin1983} is a powerful method to quantize theories with local symmetries. It is a generalization of the BRST method \cite{BRS1974b, BRS1974a, BRS1975}, first developed in the context of perturbative Yang-Mills theories, to arbitrary Lagrangian gauge theories. 
The space of on-shell, gauge invariant observables is constructed as the 0-th order cohomology of the nilpotent BV operator $s$ \cite{Barnich2000}.

The rigorous treatment of the BV formalism for infinite dimensional configuration spaces was given in the context of pAQFT in \cite{FR, Fredenhagen2013, Rejzner2015,  Rejzner2011}  (for an
alternative approach developed in Euclidean signature using factorisation algebras see \cite{CoGw-1, CoGw-2, Costello2007}). Here, we follow that general construction, introducing generating functionals for correlation functions of gauge-fixed theories in the standard way \cite{Weinberg-vol2}. Following \cite{DDPR2022, Fehre2021}, we introduce a local, momentum-independent cut-off function $Q_k$. As opposed to using a non-local regulator, a local $Q_k$ preserves the unitarity of the $S-$matrix and the structure of the propagator. Since such regulator does not deal with the UV divergences of the flow equation, these have to be treated separately: in our approach, they are regularised via a Hadamard subtraction.

Throughout this paper we avoid any reference to a particular state, since we are only interested in the algebraic structure of gauge theories with the addition of a local regulator, but not on practical computations of the flow.
Our definitions for the generating functionals, and the flow equation for the effective average action, are all given purely at the algebraic level. We only restrict the class of admissible states to that of Hadamard form, in order to have a well-defined UV renormalization of the flow equation. Explicit examples will be studied in our future work.

The main goal of this paper is to characterise the symmetries of the effective average action by introducing a constraint equation that can be used alongside our Lorentzian flow equation. Although we are mostly interested in Yang-Mills type theories and gravity, the BV formalism allows also for a generalization to gauge theories with ``open algebras'', e.g. supergravity.
The main novelty is the introduction of a new field that, together with the integral kernel of the cut-off regulator, forms a contractible pair in the cohomology of the BV operator. This is analogous to the introduction of the antighost field, forming the trivial pair with the Nakanishi-Lautrup field (the Lagrange multiplier in the gauge-fixing term of the extended action). Thanks to this field, the symmetry constraint on the effective average action takes the same form as in the non-regularised case, of the standard Slavnov-Taylor identity in Zinn-Justin form \cite{Zinn-Justin1975}.

A construction of the regulator term as a gauge-fixing term, close in spirit but different in the details, can be found in \cite{Asnafi2018}.

The symmetry constraint on the effective average action, which we call \textit{effective master equation} \eqref{eq:master ren}, can be interpreted in terms of an \textit{effective} BV formalism, in the space of effective fields $\phi$ and effective BRST sources $\sigma$. The effective average action can then be decomposed as
\[
\Gamma_k = I_{ext}(\phi, \sigma) + \hbar \hat \Gamma_k \ ,
\]
where $I_{ext}$ is the classical action, extended to include antifields, ghosts, and a generalized gauge-fixing term which enforces both the standard gauge-fixing and the regulator term. The "quantum correction" $\hat\Gamma_k$ is in the cohomology of a nilpotent operator $\Slav$, playing the role of an effective BV operator in the space of functionals of $(\phi, \sigma)$. The construction of admissible effective average actions $\Gamma_k$ then is interpreted as a cohomological problem, as in the non-regularised case \cite{Barnich2000}.

In more practical terms, this means that the effective master equation \eqref{eq:master ren} can be solved by cohomological means, controlling the functional form of $\Gamma_k$, and providing the starting point for the Ansatz to be used in the solution of the flow equation.

The paper is organised as follows. In section \ref{sec:BV}, we review the construction of the BV formalism in pAQFT. In section \ref{sec:generating functionals}, we introduce the generating functionals for the Green's functions of the theory, the regulator and gauge-fixing terms, and the effective average action. In section \ref{sec:flow-eqs} we briefly derive the flow equations following \cite{DDPR2022}. Section \ref{sec:symmetries} forms the core of the paper: here we show how to derive a Zinn-Justin equation for the effective average action by enlarging the space of fields, and we give a cohomology interpretation of the effective master equation in terms of the effective BV operator $\Slav$. In section \ref{sec:consistency} we prove that the effective master equation is consistent with the flow, studying the RG flow of composite operators in the Lorentzian case, generalising the Euclidean one \cite{Pagani2016}.

\section{The functional approach to pAQFT} \label{sec:BV}

\subsection{Kinematics}\label{sec:kinematics}
In this section we briefly review the BV formalism in pAQFT, mainly to set the notation and provide the important definitions. We refer to \cite{FR, Fredenhagen2013, Hollands2008, Rejzner2015} for the papers in which the formalism has been developed and thoroughly  discussed. In particular, in \cite{FR,Fredenhagen2013} the role of categorical concepts is emphasised in order to make the general covariance of the theory manifest. The functional approach to pAQFT and the causal renormalization procedure have been thoroughly discussed in \cite{Brunetti2009}. For a more complete introduction to these subjects, one can see the book \cite{Rejzner2016} and references therein.

Let $\mathcal M$ be an oriented, time-oriented globally hyperbolic spacetime \cite{HawkingEllis73}. The physical content of the theory is specified by a field configuration space $\conf(\mathcal M)$, which we assume to be the space of smooth sections of some natural vector bundle with fibre $V$ over $\mathcal M$: $\conf(\mathcal M) = \Gamma(\mathcal M , V)$. Compactly supported configurations are denoted by  $\conf_0(\mathcal M) = \Gamma_0(\mathcal M , V)$.

The field configuration space can be equipped with the natural Fréchet topology of uniform convergence of all the derivatives on compact sets. Using coordinate charts on $\mathcal M$ it is sufficient to specify the topology of $\mathscr E(\Omega)$ on an open subset $\Omega \subset \mathbb R^n$; this is generated by the family of seminorms \cite{Rejzner2016}
\[
\norm{\varphi}_{j,K} := \sup_{\abs{\alpha} \leq j,x}\abs{\partial^\alpha \varphi(x)} \ ,
\]
where $K\subset \Omega$ is compact and $\alpha \in \mathbb N^N$ is a multi-index, $N \in \mathbb N$. 
The topology on compactly supported field configurations $\conf_0(\Omega)$, $\Omega \subset \mathbb R^n$ can be defined as an inductive limit topology, using the fact that $\Omega$ can be exhausted by a family $K_1\subset K_2\subset\dots $ of compact sets, and each $\conf(K_n)$ is equipped with the Frech\'et topology as above. For more detail on this construction see \cite{Bourbaki} and \cite{Rudin} for the explicit construction of the seminorms. Another useful discussion of the topology on the test function space can be found in chapter 10 of \cite{Voigt}.

Classical observables are smooth maps (in the sense of Bastiani \cite{Bastiani1964}) from the configuration space to real numbers, i.e. functionals $F: \conf(\mathcal M) \to \mathbb R$. For the purpose of quantisation, we will also consider complex-valued functionals.

We now recall the relevant definitions for the functionals of physical interest.
\begin{definition}
The \textit{support} of a functional $F$ is
\begin{multline}
\supp F = \{ x \in \mathcal M \ | \ \forall \ \text{neighbourhoods} \ U \ \text{of} \ x \ \\
 \exists \varphi_1 , \ \varphi_2 \in \conf(\mathcal M) , \ \supp \varphi_2 \in U \ | \ F(\varphi_1 + \varphi_2) \neq F(\varphi_1) \}
\end{multline}
\end{definition}

\begin{definition}\label{df:Bastiani}
The \textit{derivative} of $F$ at $\varphi \in \conf(\mathcal M)$ in the direction $\psi\in \conf(\Mcal)$ is defined by
\begin{equation}
\langle F^{(1)} , \psi \rangle := \lim_{t \to 0} \frac{1}{t} \left ( F(\varphi + t \psi) - F(\varphi) \right )
\end{equation}
A functional is called continuously differentiable in an open neighbourhood $U$ of $\ph$, if  $F^{(1)}$ exists for all points in $U$ and all directions $\psi$ and the map $F^{(1)}:\conf\times\conf\rightarrow \RR$ is continuous in the product topology.
\end{definition}

\begin{definition}
A \textit{local functional} is a functional of the form
\begin{equation}
F(\varphi) = \int_x \alpha(x , \varphi, \partial \varphi , ...) \ ,
\end{equation}
for some function $\alpha$ that depends on $\varphi$ and its derivatives up to some fixed order. We denote the space of local functionals by $\Floc(\Mcal)$. 
\end{definition}
In other words, a local functional can be written as an integral of a local density, depending on the jet of $\varphi$ at the point $x$. In the above definition, we introduced the notation $\int_x = \int \dd \mu(x)$, where $\dd \mu(x)$ is the invariant measure on $\mathcal M$ induced by the metric.
\begin{definition}
    A \textit{multilocal functional} is a finite sum of functionals in the form
\begin{equation}
F(\varphi) = \int_{x_1} \ldots \int_{x_n} \beta(x_1,\dots,x_n; \varphi(x_1),\dots, \varphi(x_n), \partial \varphi(x_1),\dots \partial \varphi(x_n), \dots) \ ,
\end{equation}
up to some finite order in derivatives. We denote the space of multilocal functionals by $\F(\Mcal)$.
\end{definition}


The space of functionals can be equipped with a natural weak topology induced from the natural topology of distributions: we say that a sequence $F_n\in \mathscr F_{\mu c}$ converges to $F \in \mathscr F_{\mu c}$ in the limit $n \to \infty$ if $F_n^{(m)}(\varphi)$ converges to $F^{(m)}(\varphi)$ in $\Gamma(V^{\boxtimes m}\rightarrow  M^m)'$, in the \textit{H\"ormander topology} (see Definition~4.12 in \cite{Rejzner2016}) for every $n$ and for every field configuration $\varphi \in \conf(\mathcal M)$.

$\F(\Mcal)$ is an algebra with respect to the pointwise product:
\[
FG(\varphi) := F(\varphi) G(\varphi) \,.
\]
Vector fields on $\conf(\mathcal M)$ are smooth maps from $\conf(\mathcal M)$ to $\conf_0(\mathcal M)$. We consider the space of multilocal, smooth, compactly supported vector fields $\mathscr V(\mathcal M)$. They act on $\mathscr F(\mathcal M)$ as derivations,
\begin{equation}
\partial_X F(\varphi) := \langle F^{(1)}(\varphi) , X(\varphi) \rangle \ .
\end{equation}

The action of vector fields as derivations and the Lie bracket between two vector fields can be generalized to the \textit{Schouten bracket} between \textit{alternating multivector fields}, defined as smooth, compactly supported multilocal maps from $\conf(\mathcal M)$ to $\Lambda \conf^* (\mathcal M)' = \bigoplus \Lambda^n \conf^* (\mathcal M)'$, where
$$\Lambda^n \conf^* (\mathcal M)'\subset \Gamma((V^*)^{\boxtimes n}\rightarrow  M^n)'\,,$$  by a slight abuse of notation, denotes the space of compactly supported distributional sections which are totally antisymmetric under permutations of their arguments. Here $\boxtimes$ is the exterior tensor products of vector bundles and $V^*$ is the bundle dual to $V$. We set $\Lambda^0 \conf_0(\mathcal M) \equiv \mathbb R$ (for more details, see section 3.4 of \cite{Rejzner2016})  The space of alternating multivector fields forms a graded commutative algebra $\Lambda \mathscr V(\mathcal M)$ with respect to the product $X \wedge Y (\varphi) = X(\varphi) \wedge Y(\varphi)$. 

The Schouten bracket is an odd Poisson bracket on this algebra
\[
\{ \cdot , \cdot \} : \Lambda^n \mathscr V(\mathcal M) \times \Lambda^m \mathscr V(\mathcal M) \to \Lambda^{m+n-1} \mathscr V(\mathcal M) \ .
\]

It satisfies the following properties:
\begin{enumerate}
\item Graded antisymmetry: $\{Y, X \} = -(-1)^{(n-1)(m-1)} \{X,Y \}$ ;
\item Graded Leibniz rule: $\{X , Y \wedge Z \} = \{ X , Y \} \wedge Z + (-1)^{nm}\{X, Z \} \wedge Y$;
\item Graded Jacobi rule: \[ 
\{X , \{Y, Z \}\} - (-1)^{(n-1)(m-1)} \{ Y , \{ X, Z \} \} = \{ \{ X , Y \} , Z \} \ .
\]
\end{enumerate}

In the standard approach to the BV formalism, vector fields are identified with \textit{antifields}. The action of vector fields on classical observables can be formally written as
\begin{equation}
\partial_X F(\varphi) = \int_x X(\varphi) \frac{\delta F(\varphi)}{\delta \varphi(x)} \ .
\end{equation}
Identifying the functional derivatives $\frac{\delta}{\delta \varphi}$ with the antifields $\varphi^\ddagger$, we see that the algebra of alternating multivector fields is generated by fields and antifields, and the Schouten bracket between two multivector fields $X,Y$, with degree respectively $n, \ m$, is interpreted in the physics literature as the \textit{antibracket}:
\begin{equation}
\{ X , Y \} = - \int_x \left ( \frac{\delta X}{\delta \varphi(x)} \frac{\delta Y}{\delta \varphi^\ddagger(x)} + (-1)^{n} \frac{\delta X}{\delta \varphi^\ddagger(x)} \frac{\delta Y}{\delta \varphi(x)} \right ) \ .
\end{equation}

\subsection{Dynamics and symmetries}\label{sec:dynamics}
To introduce a variational principle in the theory, we define a \textit{Lagrangian} $L$ as a natural transformation between the functor of compactly supported smooth functions spaces $\mathcal D: \mathbf{Loc} \to \mathbf{Vec}$ and the functor of $\mathscr F_{\text{loc}}:\mathbf{Loc} \to \mathbf{Vec}$, satisfying $\supp(L(f)) \subseteq \supp(f)$ and the \textit{additivity rule}
\begin{equation}
L(f+g+h)  = L(f+g) + L(g+h) - L(g) \ , \ f, g, h \in \mathcal D(\mathcal M) \ ,
\end{equation}
where $\supp f \cap \supp h = \emptyset$. The action $I(L)$ is defined as an equivalence class of Lagrangians, where two Lagrangians are equivalent if
\begin{equation} \label{eq:equivalence relation between lagrangians}
L_1 \sim L_2 \Leftrightarrow \supp(L_1 - L_2)(f) \subset \supp \dd f \ .
\end{equation}

By definition, the Lagrangian takes an element $f \in \mathcal D(\mathcal M)$ and maps it into the space of local functionals, which depend on the field configurations $\varphi$, which we then denote the notation $L(f)$. When we also need to emphasise the dependence on $\varphi$, we use the notation $L(f)[\varphi]$. We do not explicitly require differentiability with respect to the test function $f$, but in the examples we consider in this work, the Lagrangians $L(f)$ will in fact be differentiable or even smooth with respect to $f$. 


We do not assume differentiability in f, and we also do not use it later. However, it does hold in many examples. In particular, our default way of constructing generalised Lagrangians is to multiply fields by appropriate test functions (as explained below). In this case, the Lagrangians are smooth with respect to those test functions, as they are smooth with respect to the fields.

The definition in terms of natural transformations, first introduced by Brunetti, D\"utsch, and Fredenhagen \cite{Brunetti1999}, is needed because neither the spacetime nor the support of the field configurations are typically compact, so we cannot simply define the action as an integral over $\mathcal M$. We will comment in Section \ref{subsec:IR regularization} on the explicit construction of a generalised Lagrangian from typical Lagrangian densities of QFT, such as the $\phi^4$ or the Yang-Mills type theories.

From now on, we will drop $\mathcal M$ from the notation, since we always work on a fixed spacetime $\mathcal M$.

The \textit{Euler-Lagrange} derivative of $I$ is defined by
\begin{equation}
\langle I^{(1)}(\varphi) , g \rangle := \langle L^{(1)}(f)(\varphi) , g \rangle \ ,
\end{equation}
and the field equations are
\begin{equation} \label{eq:EOM}
I^{(1)}(\varphi) = 0 \ .
\end{equation}
	
The space of solutions to \eqref{eq:EOM} is denoted by $\conf_{os} $, and multilocal functionals on this subspace are called \textit{on-shell} functionals, $\mathscr F_{\text{os}}$. This space can be characterised \cite{Rejzner2011} as the quotient space $\mathscr F_{\text{os}} = \mathscr F / \mathscr F_{0}$, where $\mathscr F_{0}$ is the ideal generated by the equations of motion. We assume that this ideal coincides with the space of functionals that vanish on-shell, which is the case for theories of physical interest. For the precise statement of appropriate regularity conditions, see e.g. \cite{henneaux1990lectures}.

The BV formalism aims at a homological description of on-shell, gauge-invariant functionals. We can now take the first step in the BV construction, giving a geometrical description of on-shell functionals. In fact, these are characterised by the zeroth homology of a certain differential, called the \textit{Koszul map}.
Given an action $I$, the \textit{Koszul map} is defined as the operator acting on alternating multivector fields $X \in \Lambda \mathscr V $ as
\begin{equation}\label{def:Koszul-map}
\delta_K(X) := \{ X , I \} \ ,\  X \in \Lambda \mathscr V  \ .
\end{equation}

The Koszul map $\delta_{K}$ is a differential, i.e. $\delta_{K}^2=0$ and the exterior algebra $\Lambda \mathscr V $ is a graded vector space (with the grading given by the exterior power). The differential increases the degree by one and we can write down the following complex:
\begin{equation}
\ldots \xrightarrow{\delta_{K}^{(-3)}} \underset{-2}{\Lambda^2 \mathscr V } \xrightarrow{\delta_{K}^{(-2)}} \underset{-1}{\Lambda^1 \mathscr V } \xrightarrow{\delta_{K}^{(-1)}} \underset{0}{\mathscr F  } \xrightarrow{\delta_{K}^{(0)}} 0 \ .
\end{equation}

Here we use the superscript in $\delta_{K}^{(n)}$ to indicate the degree of the vector spaces where it is acting. The $n$-th cohomology of this complex is defined as 
\[
H^n(\delta_K):=\frac{{\rm Ker}\, \delta_k^{(n)}}{{\rm Im}\, \delta_k^{(n-1)}}
\]

The image of $\delta^{(-1)}_K$ is contained in the space of functionals that vanish on-shell, $\mathscr F_0$. Under the appropriate regularity conditions \cite{henneaux1990lectures} (when $\mathscr F_0$ coincides with the ideal generated by the equations of motion), it is possible to prove that $\text{Im}\, \delta_K = \mathscr F_0$ \cite{henneaux1990lectures}, and the space of on-shell functionals $\mathscr F_{\text{os}}$ is characterised by the 0th cohomology of the Koszul operator $H^0(\delta_K)=\mathscr F/\mathscr F_0$ \cite{Rejzner2011}. For more on homological algebra, see \cite{Weibel}.

Next we discuss symmetries. A vector field $X\in\V$ is a symmetry of a Lagrangian $L$ if the corresponding Euler-Lagrangian derivative $I^{(1)}$ of $I(L)$ satisfies $\langle I^{(1)}, X \rangle\equiv 0$. Symmetries form a Lie-subalgebra of $\V$ and we assume that each symmetry can be written as $X=\rho(\xi)+X_0$, where $X_0$ is a symmetry vanishing on-shell, $\xi\in \g$, where $\g$ is a Lie-algebra that can be expressed as a space of smooth sections of some vector bundle $\mathscr{k}\xrightarrow{\pi}\Mcal$ and $\rho:\g\rightarrow \V$ is a Lie-algebra morphism. Physical examples include Yang-Mills theories and gravity with the Einstein-Hilbert action.

The space of functionals invariant under the action of $\g$
has a cohomological description in terms of the \textit{Chevalley-Eilenberg cochain complex}. Its underlying graded algebra is the space of multilocal maps from $\conf$ to $\Lambda \mathscr g'$, where $\Lambda^n \mathscr g'$ denotes the dual of $\Gamma(\mathscr{k}^{\boxtimes n}\rightarrow \Mcal^n)$, consisting of distributions totally antisymmetric in their arguments.  The grading is called the \textit{pure ghost number} $\# \text{pg}$, and the forms are known in the physics literature as \textit{ghosts}.

The Chevalley-Eilenberg differential $\gamma_{ce}$ in degree zero is given by $\gamma_{ce}F(\xi)=\partial_{\rho(\xi)} F$, in degree one it acts through the Lie bracket
and on a general functional of gauge fields and ghosts by
\begin{multline*}
\gamma_{ce}F(A; c_0, \ldots, c_q) := \sum_{i=0}^q (-1)^{i}\partial_\rho(c_i)(\iota_{(c_0,\ldots,c_{i-1}, c_{i+1},\ldots, c_q}F)(A) \\
+\sum_{i<j}^q (-1)^{i+j}F(A,[c_i,c_j],\ldots,c_{i-1},c_{i+1},\ldots,c_{j-1},c_{j+1},\ldots,c_q) \ ,
\end{multline*}
where $\iota$ denotes the insertion of $n$ vector fields on an $n$ form. In particular, for $F \in \mathscr F(\mathscr E)$ it holds
\begin{equation} \label{eq:gauge-transformation-functional-a}
\gamma_{ce} F(\A) = \langle F^{(1)} , D_\A c \rangle \ .
\end{equation} 

The Chevalley-Eilenberg differential generates the Chevalley-Eilenberg cochain complex,
\[
0 \xrightarrow{\gamma_{ce}^{(-1)}}  \underset{0}{\mathscr F(\mathscr E)} \xrightarrow{\gamma_{ce}^{(0)}} \underset{1}{C^\infty(\mathscr E, \Lambda^1 \g[1]) }\xrightarrow{\gamma_{ce}^{(1)}} \underset{2}{\mathscr{C}^\infty(\mathscr E , \Lambda^2 \g[1])} \ldots \ .
\]

From the action of the Chevalley-Eilenberg differential on a functional of the field configuration space, Eq. \eqref{eq:gauge-transformation-functional-a}, we see that the kernel of $\gamma_{ce}^{(0)}$  characterises the gauge-invariant functionals. On the other hand, the image of $\gamma_{ce}^{(-1)}$  is simply the $0$ functional. It follows that the cohomology in degree $0$ of the Chevalley-Eilenberg complex describes the gauge-invariant functionals, $H^0(\gamma_{ce}) = \mathscr F^{\text{inv}}$.


In order to describe the space of gauge-invariant functionals on-shell, we use the Batalin-Vilkovisky graded algebra $\mathscr{BV}$, defined as the space of compactly supported multilocal multi-vector fields on the extended configuration space $\overline{ \conf} := \conf[0] \oplus \mathscr g[1]$, where the number in brackets denotes the pure ghost number. This space can be equipped with a locally complex topology, which is the appropriate variant of the H\"ormander topology introduced in Section~\ref{sec:kinematics}. See \cite[Chapter 3]{Rejzner2016} for more details.

We denote the ghost fields as $c \in \mathscr g[1]$, and the field configuration multiplet, element of the extended configuration space, with $\varphi := \{ \varphi, c\} \in \overline{\conf}$. $\mathscr{BV}$ can also be seen as the the space of multilocal functionals on the odd cotangent bundle of $\overline{ \conf}$. The elements of the cotangent space are the antifields for the fields and the ghosts, and we identify the functional derivatives $\varphi^\ddagger := \frac{\delta}{\delta \varphi}$ as the "basis" for the fibre $\conf^*[-1] \oplus \mathscr g^*[-2]$, where $*$ means taking sections of the corresponding dual bundle. The grading on 
fibre is related to the \emph{antifield number} $\#af$ (equal to minus the number in square brackets).

The functional-analytic considerations of Section~\ref{sec:kinematics} apply also to the extended configurations space, but one needs to keep track of the grading. Local and multilocal functionals are defined completely analogously. For more detail, see Chapter~3.3 of \cite{Rejzner2016}. For simplicity of notation, we will now use $\conf$ to denote the extended configuration space.

The $\mathscr{BV}$ algebra has two gradings, the ghost number $\#gh$ and the antifield number $\#af$, related to the pure ghost number via $\#gh = \#pg - \#af$. Seen as the space of graded multivector fields, $\mathscr{BV}$ is equipped with a graded Schouten bracket, which we also denote by $\{ \ . \ , \ . \ \}$.

The BV differential is defined as
\begin{equation}\label{eq:BVanti}
s_{BV}X:=\{X,L(f)+L_{ce}(f)\}\,,
\end{equation}
where $f\equiv 1$ on the causal completion of the support of $X$ and
$L_{ce}$ is obtained from $\gamma_{ce}$ by introducing a cut-off function (to ensure compact support). In what follows we will use the notation $\{X,I+I_{ce}\}$ when we mean the bracket with the Lagrangian for the  choice of $f$ as above. In particular, we introduce $f$ in $L_{ce}$ in the natural way, multiplying the ghosts with $f$. 

A crucial requirement of the BV formalism is that the BV operator $s_{BV}$ is an actual differential, $s^2_{BV}$. 
To ensure the nilpotency of $s_{BV}$, a sufficient condition is to choose a Lagrangian satisfying the \textit{Classical Master Equation} (CME) modulo boundary terms, which takes the expression
\begin{equation}
\{ L_{cl}(f) , L_{cl}(f) \}\sim 0 \,,
\end{equation}
where the equivalence relation is defined in Eq. \eqref{eq:equivalence relation between lagrangians}. For Yang-Mills theories and gravity, the classical Lagrangian $L_{cl}(f)$ is simply $L_{cl}(f) = L(f) + L_{ce}(f)$, as it satisfies the CME (\textit{cf.} the example \ref{sec:YM-example-classical}), but sometimes additional terms, with a higher number of antifields, are needed \cite{Gomis1994}.


The BV differential can be expanded in antifield number into:
\begin{equation}
s_{BV} = \delta_K + \gamma_{ce} + ... \ .
\end{equation}
In the formula above, the dots represent the fact that, in the most general case, $\delta_K + \gamma_{ce}$ fails to be a nilpotent operator. In the case of Yang-Mills theories and gravity, the sum contains the first two terms only.

The cohomology of $\gamma_{ce}$ characterises invariant functionals, and the homology of $\delta_K$ the on-shell functionals. The main theorem of homological perturbation theory guarantees that the gauge-invariant, on-shell observables of the theory can be characterised by the zeroth cohomology of the BV differential, $\mathscr F^{\text{inv}}_{\text{os}} = H^0(\mathscr{BV} , s_{BV})$ \cite{Barnich2000, FR, Fredenhagen2013, Rejzner2011}. In fact, it can be proven that
\begin{equation}\label{eq:main theorem of hom pt}
H^0(\mathscr{BV}, s_{BV}) = H^0(H^0( \mathscr{BV} , \delta_K ) , \gamma_{ce}) \ .
\end{equation}


\subsection{IR regularization and the adiabatic limit} \label{subsec:IR regularization}
In this section, we briefly comment on the construction of generalized Lagrangians. The motivation is the basic observation that, since a globally hyperbolic spacetime cannot be compact, an action integral $I = \int_x L(\varphi)$ is necessarily divergent, so we introduce the cutoff $f$ to make it finite. Some computations then depend on the support of $f$, and it is often necessary to take the limit in which the support of $f$ is extended over the whole spacetime $\mathcal M$.

In particular, the flow equations in \cite{DDPR2022} depend on $f$, but the adiabatic limit can be taken on the level of
the beta functions, obtained through functional derivatives of the flow equation.

More precisely, assuming without loss of generality that $f=1$ in a neighbourhood of an origin of the spacetime, the adiabatic limit can be taken replacing $f$ with $f(x/n)$, with $n \in \mathbb N$, and eventually considering the limit $n \to \infty$ for the quantities of interest. 

A straightforward procedure to build a generalized Lagrangian from the Lagrangian density of some theory is to consider
\[
I = \int_x f L \ ,
\]
for some test function $f$. The test function is assumed to be equal to $1$ in some finite region of spacetime, so that the action functional is well-defined. 

However, with this choice of regularization the Quantum Master Equation \eqref{eq:QME} for Yang-Mills theory does not hold and it is the necessary condition to ensure that the $S-$matrix is gauge independent.

Hence, we choose a different regularisation for the generalised Lagrangian. We introduce a pair $f=(f_\varphi,  f_c)$ of test functions, and, given a particular Lagrangian density, we define a generalised Lagrangian as
\begin{equation}
I(f_\varphi,  f_c) = \int_x L(f_\varphi \varphi,  f_c c) \ .
\end{equation}
In other words, we introduce a set of regularised, compactly supported variables $(f_\varphi \varphi,  f_c c)$, substituting the original gauge fields $\varphi$ and ghosts $c$.
We choose  $f_\varphi$ and $ f_c$ so that $f_\varphi\equiv 1$
on the support of $ f_c$. We also regularise the antifields. We replace $\frac{\delta}{\delta \ph(x)}$ with 
$\frac{\delta}{\delta f_\varphi\ph(x)}$, which is a shorthand for taking a derivative and dividing the result by $f_\varphi$. Similarly, we replace $\frac{\delta}{\delta c(x)}$ with 
$\frac{\delta}{\delta  f_c c(x)}$. 

For such choice of test functions, the classical master equation can be satisfied exactly (not just modulo the equivalence relation on the space of generalized Lagrangians).

For example, the Yang-Mills Lagrangian regularised by multiplying the vector potential $A$ with the appropriate test function $f_A$ is invariant under the regularised gauge transformation, where one replaces the ghosts $c$ (corresponding to gauge parameters) with $c  f_c$, since 
$f_A\equiv 1$ on the support of $ f_c$, by assumption. As a consequence, the corresponding BV-extended regularised  Lagrangian satisfies the CME in the strict sense. We discuss this in detail in Section~\ref{sec:YM-example-classical}.

In section \ref{sec:mSTI}, we will show explicitly that, for Yang-Mills-type theories (that are, both possibly non-abelian gauge theories as QED and QCD, and gravity), also the Quantum Master Eqaution (QME) is exactly satisfied by a generalized Lagrangian with this choice of regularization.

In \eqref{eq:BVanti}, the BV differential is given as the antibracket with the extended action, $(f_\varphi \varphi,  f_c c)$ and we choose  $f_\varphi$ and $ f_c$ so that $f_\varphi =  f_c = 1$ on the support of $X$. Hence 
the test functions do not influence the gauge transformation properties of $F$.

With the above choice of regularisation, we can identify a local functional with a generalized Lagrangian $f\mapsto F(
\chi f)$. In the following, for notational convenience, we will use the same symbol for the functional and the corresponding generalized Lagrangian.

	\subsection{Gauge fixing and non-minimal sector}
In the BV formalism, the gauge fixing is performed in two steps. First, we need to enlarge the BV complex to include the \textit{antighosts} and the \textit{Nakanishi-Lautrup fields}. These form a contractible pair, i.e., $s_{BV} \bar c = i b$ and $s_{BV} b = 0$, so they do not contribute to the cohomology of $s_{BV}$ \cite{Barnich2000}\footnote{Note that the authors of \cite{Barnich2000} call BRST differential what we call the BV differential. We reserve the term ``BRST differential'' for the specific part of $s$ that coincides with the differential proposed by Becchi, Rouet, and Stora.}. Together with $\bar c$ and $b$, we introduce the corresponding antifields. For notational simplicity, we denote the new space of field configurations $\conf$, so that now the field multiplet $\varphi$ includes the antighosts and the Nakanishi-Lautrup fields too, and the extended BV complex is again denoted by $\mathscr{BV}$.

The gauge-fixing is performed as an automorphism of the BV complex, leaving the antibracket invariant, and such that the transformed part of the action that does not contain antifields has a well posed Cauchy problem. First, we introduce the gauge-fixing Fermion $\Psi$ as a fixed generalized Lagrangian with $\#gh = -1$ and $\#af = 0$.
We then define the automorphism by setting
\begin{equation}\label{eq:alpha def}
\alpha_\Psi(F) = F \left (\varphi,c,\varphi^\ddagger + \frac{\delta \Psi (f_\varphi,f_c)}{\delta f_{\varphi} \varphi}, c^\ddagger + \frac{\delta \Psi (f_\varphi,f_c)}{\delta f_cc}  \right) \ ,
\end{equation}
This means that in all expressions we replace $\varphi^\ddagger$ with $\varphi^\ddagger +  \frac{\delta \Psi}{\delta f_{\varphi}\varphi}$. Similarly for $c^\ddagger $. This change of variables makes sense, since this derivative is smooth (recall that the first derivative of a local functional is always smooth).
It can be checked that this is a canonical transformation with respect to the antibracket. See \cite{Gomis1994,FR} for more details.

 The gauge-fixed Lagrangian is now $L_{ext} := \alpha_\Psi(L_{cl})$. For theories that are linear in the antifields, as Yang-Mills and gravity, the gauge-fixed action $I_{ext}$ takes the usual form (here we ommitted the test functions from the notation, for clarity)
\begin{equation}
I_{ext}:= \alpha_\Psi(I_{cl}) = I + I_{af} + s_{BV} \Psi = I + I_{af} + I_{gh} + I_{gf} \ ,
\end{equation}
where $I$ is the original, gauge-invariant action and the antifield, ghost, and gauge-fixing sectors are respectively $I_{af}$, $I_{gh}$, and $I_{gf}$. Note that $I_{af}$ contains $I_{ce}$ and a term arising from introducing the non-minimal sector, i.e $\int_x  if_cb\frac{\delta}{\delta f_A \bar{c}}$.

The gauge-fixed BV differential $s = \alpha_\psi \circ s_{BV} \circ \alpha_\psi^{-1}$ can be expanded in \emph{total antifield number} ($\#ta$ of any antifield is one) into
\begin{equation}
s = \{ \ \cdot \ , I_{ext} \} = \delta + \BRST \ .
\end{equation}
$\delta$ is the Koszul operator for the field equations derived from $I+ I_{gh} + I_{gf}$, and $\gamma_{BRST}$ is the gauge-fixed BRST operator. The space of classical observables  is now isomorphic to $H^0(\mathscr{BV} , s)$, where the isomorphism is given by $\alpha_\psi$. 

Using the main theorem of homological perturbation theory, it is possible to write the 0th-cohomology of $s$ in terms of $\delta$ and $\BRST$ as
\begin{equation}
H^0(\mathscr{BV},) = H^0(H^0( \mathscr{BV} , \delta ) , \BRST) \\,,
\end{equation}
so it is exactly the space of BRST-invariant observables on-shell.

In the following, we will assume that the gauge fixing has been already performed and will denote $I_{ext}$ simply as $I$.

\subsubsection{Example: Yang-Mills theory} 
\label{sec:YM-example-classical}

Let's discuss the BV construction in the concrete example of Yang-Mills (YM) theories \cite{Hollands2008}. These are defined by a globally hyperbolic spacetime $(\mathcal M, g)$, and a principal bundle $P \to \mathcal M$, where $G = SU(N)$ is a compact Lie group. In a local trivialisation of the bundle, a connection $\A$ is a $\g -$valued 1-form, and the Yang-Mills configuration space $\mathscr E = \Omega_1 (\mathcal M , \mathscr g)$ are 1-forms with values in the Lie algebra $\mathscr g = \mathscr {su}(N)$, for some $N$. The fundamental invariant associated with the connection $\A$ is the curvature $F := \dd \A + \lambda_{YM} \frac{1}{2} [ \A , \A ]$, where $\lambda_{YM}$ is the YM coupling constant, and the invariant action for Yang-Mills theories can be written in terms of the following generalized Lagrangian:
\[
I_{inv}(f_\A) = - \frac{1}{2} \int_x \Tr{F(f_\A \A) \wedge \star F(f_\A \A) } \ ,
\]
where $\star$ is the Hodge operator and $\Tr$ is the trace in the adjoint representation of $\mathscr{su}(N)$ defined by the Killing-Cartan metric. 

The action $I_{inv}$ is invariant under the action of a $G$-valued function $g$ of the form
\[
\A \to \A_g = - \frac{i}{\lambda_{YM}} g \dd(g^{-1}) + g \A g^{-1} \ ,
\]
which infinitesimally takes the form of standard gauge transformations
\[
\delta_\xi \A = D_\A \xi = \dd \xi + i \lambda_{YM} [\A, \xi] ,
\]
where $\xi$ is a gauge parameter, i.e. $\xi\in \mathscr{C}^{\infty}(\mathcal M , \mathscr su (N))$.
$D_\A$ is the covariant extension of the exterior derivative, which acts as $D_\A = \dd + i \lambda_{YM} \A$ on the fundamental representation, where the action of $\A$ is by multiplication, and as in the formula above in the adjoint representation of the gauge group $SU(N)$ on its Lie algebra, that is, on the gauge field itself. Indeed, for constant gauge transformations (that is, for elements of the gauge group itself), the gauge field transforms in the adjoint representation as $\A \to \A_g = g \A g^{-1}$.

The EOM derived from the invariant YM action are
\begin{equation} \label{eq:EOM-YM}
D_\A \star F = 0 \ .
\end{equation}

The Chevalley-Eilenberg complex is the algebra $\mathscr {CE} = \mathscr F(\mathscr E \oplus \g[1])$ of microcausal functionals of the gauge fields $\A$ and ghosts $c$. The Chevalley-Eilenberg differential is defined by its action on the fields as
\[
\gamma_{ce} \A = D_\A c \ , \quad \gamma_{ce} c =  - \frac{i \lambda_{YM}}{2} [c , c] \ .
\]
The BV algebra is the algebra of microcausal functionals on the odd cotangent bundle of the extended configuration space $T^*(\overline{\mathscr E})$. Identifying the elements of the tangent space with antifields, their basis is given by the derivatives $\A^\ddagger = \frac{\delta}{\delta \A}$ and $c^\ddagger = \frac{\delta}{\delta c}$. The Koszul map now by definition acts on a generic functional $X$ as
\[
\delta_K (X) = \{ X , I_{inv} \} = \int_x \frac{\delta X}{\delta \A^{\ddagger}} D_\A \star F \ .
\]

The Chevalley-Eilenberg contribution to the action follows from the action of the Chevalley-Eilenberg differential, and it is given by
\[
I_{ce}(f_A,f_{c}) = \int_x \left ( \A^\ddagger D_\A(f_c c )  - \frac{i\lambda_{YM}}{2} c^\ddagger [ c, f_c c]   \right ) \ ,
\]
where we used the fact that $f_\A=1$ on the support of $ f_c$ and in the second term one factor $f_c$ was canceled by $1/f_c $ present in $\frac{\delta}{\delta f_c c}$.

We can now check by direct computation that the combination $I_{inv} + I_{ce}$ satisfies the CME, i.e. that $\{I_{cl}, I_{cl} \} = 0$. First we note that
\[
\{ I_{inv} + I_{ce} , I_{inv} + I_{ce} \} = 2 \{ I_{inv} , I_{ce} \} + \{ I_{ce} , I_{ce} \} \ ,
\]
since we have $\{ I_{inv} , I_{inv} \} = 0$ as $I_{inv}$ does not contain antifields; the contribution $\{ I_{inv} , I_{ce} \}$ vanishes thanks to gauge invariance of the YM action. It remains to check that the last term vanishes as well. This is given by
\[
\{ I_{ce} , I_{ce} \} = 2 \int_x \left ( \frac{\delta D_\A f_cc}{\delta \A} D_\A (f_cc) - \frac{i \lambda_{YM}}{2} \frac{\delta}{\delta f_c c} (D_\A (f_cc) + [f_cc,f_cc]) [f_cc,f_cc] \right ) \ .
\]
Now, we have
\begin{gather*}
\frac{\delta D_\A f_cc}{\delta \A } D_\A f_cc = i \lambda_{YM} [ D_\A f_cc , f_cc] \ , \ \text{and} \\
\frac{i\lambda_{YM}}{2}\frac{\delta}{\delta f_cc} (D_\A f_cc + \frac{1}{2} [f_cc,f_cc]) [f_cc,f_cc] = \frac{i\lambda_{YM}}{2} D_\A  [f_cc,f_cc] + [f_cc, [f_cc,f_cc]] \ .
\end{gather*}
Since $[f_cc,[f_cc,f_cc]] = 0$ by Jacobi identity, the two remaining terms cancel; therefore the classical action $I_{cl} = I_{inv} + I_{ce}$ satisfies the CME. 

The EOM \eqref{eq:EOM-YM} are not normally hyperbolic, since the action satisfies the Noether identities associated with gauge invariance,
\[
\int_x \frac{\delta I_{inv}}{\delta \A(x)} D_\A f_cc(x) = 0 \ .
\]
Therefore, we need to perform gauge fixing. This is an automorphism of the BV algebra that leaves the antibracket invariant, and such that its action on $I_{cl}$ produces a gauge-fixed action that satisfies the CME and have normally hyperbolic EOM. For a generic gauge-fixing function $\mathcal G(\A)$, the gauge-fixing Fermion, which is again a generalized Lagrangian, takes the form
\[
\Psi(f_A,f_c)[A,c,b,\bar{c}] = i \int_x f_A\bar c \left ( \alpha f_c \frac{b}{2} + \mathcal G(f_A\A) \right ) \ ,
\]
where $\alpha$ is a constant.
The gauge-fixed action $I = \alpha_{\psi}(I_{cl})$ is the sum of two terms only, since the classical action is linear in the antifields:
\[
I(f_A,f_c) = I(f_A,f_c) \left [\varphi , \varphi^\ddagger + \frac{\delta \Psi (f_A,f_c)}{\delta f_{\varphi} \varphi}\right] = I_{cl}(f_A,f_c) + s_{BV} \Psi(f_A,f_c) \ .
\]

It follows by direct computation that the gauge-fixed action for Yang-Mills theories takes the familiar expression
\begin{multline} \label{eq:YM-gauge-fixed-action}
I(\varphi, \varphi^\ddagger) = - \frac{1}{2} \int_x  \Tr{F(f_\A \A) \wedge \star F(f_\A \A) }\\- i \int_x \Tr [ d \bar c , D_A f_cc ] - \int_x \frac{\alpha}{2} f_c b^2 + f_c b \mathcal G(\A) +\ \textrm{antifield terms}\,,
\end{multline}
where now the field multiplet is to be understood as $\ph=(A,c,\bar{c},b)$.

In this example, we can choose in particular the gauge-fixing functional
\[
\mathcal G(\A) = \star \dd \star \A = \nabla_\mu \A^\mu \ ,
\]
known as Lorenz gauge. The linearisation of the action into $I_0 + V$ around the trivial field configuration $\varphi = 0$ produces the free wave operator $P_0$; in the case $\alpha= 1$ it can be written in terms of the Hodge laplacian $\square_H := - (\tilde \delta \dd + \dd \tilde \delta )$, where $\tilde \delta$ is the co-differential, as \cite{Rejzner2016}
\[
P_0 =
\begin{pmatrix}
\square_H + \dd \tilde \delta - \dd & 0 & 0 & 0 \\
\tilde \delta & -1 & 0 & 0 \\
0 & 0 & 0 & i \square_H \\
0 & 0 & - i \square_H \\
\end{pmatrix} \ .
\]
The expression for the advanced and retarded fundamental solutions can be found in Ref. \cite{Rejzner2016}.

\subsection{Quantum theory}
In order to quantize the theory, we split the action into $I = I_0 + V$, where $I_0$ is a term quadratic in the fields, with $\# af = 0$, and $V$ is the remaining, interacting term.
Denote $\delta_0 := \{ \ . \ , I_0 \}$.

After gauge fixing, the Euler-Lagrange derivative of $I_0$ induces a normally hyperbolic operator $P_0 : \conf \to \conf^*\subset  \conf_0'$, where the operation $*$ means that we have dualised all the bundles. Being normally hyperbolic implies that $P_0$ has unique advanced and retarded propagators $\Delta_{A,R}: \C_0^*\rightarrow \C$, satisfying
\begin{align}
P_0 \Delta_{A,R} &=\text{id}\big|_{ \conf^*_0}\,,\quad \Delta_{A,R} P_0\big|_{ \conf_0} = \text{id}\big|_{ \conf_0} \ , \\
\supp \Delta_{A,R} &\subset \{ (x,y) \in \mathcal M \times \mathcal M \ | \ y \in J^{+,-}(x) \} \ ,
\end{align}
where $J^\pm$ denotes the causal future (past) of $x$, and to the advanced propagator corresponds $J^+$ in the above formula. The \textit{causal propagator} (also called \textit{Pauli-Jordan commutator function}) is defined as the difference
\begin{equation}
    \Delta := \Delta_R - \Delta_A\,.
\end{equation}

The quantum algebra is then constructed as a deformation of the classical BV algebra \cite{Brunetti1999}, in which we introduce a non-commutative $\star-$product in the space of formal power series $\mathscr{BV}[[\hbar]]$. We first define the \textit{microcausal functionals} $\mathscr F_{\mu c} $ as the set of smooth functionals, such that the wave front set of their derivatives does not contain elements $(x_1,..., x_n; k_1,...,k_n) \ , \ k_1 \in T^*_{x_i} \mathcal M \ , \ i= 1,...,n$, where $k_i$ are covectors in the closed forward or backward light-cones. In this space, we define the non-commutative product as
\begin{equation}
A \star B := m \circ e^{\hbar\Gamma_{\Delta_+}}(A \otimes B) \ ,
\end{equation}
where $\Gamma_{\Delta_+}$ is the functional differential operator $$\Gamma_{\Delta_+} =  \int_{x,y} \Delta_+(x,y) \frac{\delta}{\delta \varphi(x)} \otimes \frac{\delta}{\delta \varphi(y)}\,,$$ and $m$ is the pointwise multiplication. 

The non-commutative product is defined in terms of a bidistribution $\Delta_+ : \mathscr E^*   \times \mathscr E^*  \to \mathbb C$, given by $\Delta_+ := \Delta_S + \frac{i}2\Delta$,  where $\Delta_S$ is a symmetric bisolution of $P_0$ chosen in such a way that $\Delta_+$ is of positive type and it satisfies the \textit{microlocal spectrum condition} on its wavefront set \cite{Radzikowski1996}:
\[
\text{WF}(\Delta_+) = \{(x,y;k_x,k_y)\in T^{*}(\mathcal M^2)\setminus \{0\} \ |\ (x,k_x)\sim(y,-k_y), k_x \triangleright 0 \} \ ,
\]
where $(x,k_x)\sim(y,-k_y)$ holds if $x$ and $y$ are joined by a null geodesic $\lambda$, $g^{-1}k_x$ is tangent to $\gamma$ at $x$ and $-k_y$ is the parallel transport of $k_x$ along $\lambda$. $k_x\triangleright 0$ holds if $g^{-1} k_x$ is future pointing.

This condition is equivalent to require that the short-distance singularity structure of the 2-point function has a universal form. For the case of a massive scalar field in even dimensions, for example, the Hadamard 2-point function must take the expression \cite{Brunetti2009}
\begin{equation}
    \Delta_+ = \frac{(-1)^{\frac{d}{2}}}{2(2\pi)^{\frac{d}{2}}} \log\frac{\mu^2}{m^2} m^{d/2-1}\sigma^{\frac{2-d}{4}} I_{d/2 -1}(\sqrt{m^2\sigma}) + w = H + w \ ,
\end{equation}

where $\sigma$ is the squared geodesic distance between $x$ and $y$, $I_d$ is the modified Bessel function of the first kind, $u$ and $v$ are two functions uniquely determined by the geometry of $\mathcal M$ and the equations of motion, $\mu$ is an arbitrary mass scale, and $w$ is the smooth part of $\Delta_+$.
Similar conditions hold in odd dimensions as well as for vector and tensor fields, where the functions $u$ and $v$ are replaced by appropriate tensor structures.

Different choices of $\Delta_S$ produce different representations of the quantum algebra; they are however isomorphic, with the isomorphism given in terms of $e^{\hbar/2 \Gamma_{\Delta_+' - \Delta_+}}$, where $\Delta_+'$ is the two-point function for some different choice of $\Delta_S$. 

Finally, we introduce an involution $*$ in the algebra of microcausal functionals, equipped with the non-commutative $\star-$product, acting by complex conjugation,
\[
F^*(\varphi) := \overline{F(\varphi)} \ .
\]
Together with the involution, the quantum $*-$algebra of off-shell functionals will be denoted as $\mathscr A := (\mathscr{BV}[[\hbar]], *, \star)$.

In order to discuss interactions, we need to introduce a second product in the algebra, the \textit{time-ordered product}. We first restrict our attention to \textit{regular functionals}, that are functionals whose derivatives are smooth sections, and we define the time-ordered product as
\begin{equation}
A \cdot_T B := m \circ e^{\hbar \Gamma_{\Delta_F}} (A \otimes B) \ ,
\end{equation}
with $\Gamma_{\Delta_F} = \int_{x,y} \Delta_F(x,y) \frac{\delta}{\delta \varphi(x)} \otimes \frac{\delta}{\delta \varphi(y)}$, in terms of the \textit{Feynman propagator} $\Delta_F := \Delta_+ + i \Delta_A$.

\begin{remark}[Feynman Hadamard parametrix]\label{rmk:Hadamard-Feynman}
In a Hadamard state, similarly to the 2-point function, the Feynman propagator has a universal, short-distance singularity structure given by the time-ordered, Feynman Hadamard parametrix, so that locally the Feynman propagator takes the form
\[
\Delta_F(x,y) = H_F(x,y) + w(x,y) \ .
\]
\end{remark}
The time-ordered product is equivalent to the pointwise product of a classical functional thanks to the existence of an invertible linear operator $T$, defined by
\begin{equation}
    T(F) := \exp{\frac{\hbar}{2} \int_{x,y} \Delta_F(x,y) \frac{\delta^2}{\delta \varphi(x) \delta \varphi(y)}}F \ .
\end{equation}
Indeed, the time-ordered product can now be written as
\[
F \cdot_T G = T(T^{-1}F \cdot T^{-1}G) \ .
\]
The equivalence holds also at the level of renormalised time-ordered products between local functionals \cite{Brunetti2009, Fredenhagen2013}, as introduced later in Section~2.6. The relevant formula is \eqref{eq:renormalized-time-ordered-product}.

Due to the structure of the wave front set of the Feynman propagator at coinciding points, the time-ordered product cannot be straightforwardly extended to microcausal functionals as the $\star-$product. This is the well-known problem of UV divergences in perturbative QFT. Its solution is provided by \textit{causal renormalization}, which we review in the next section.

Using the time-ordered product, we can define the $S-$matrix
\begin{equation}
S(V) := e^{\frac{i}{\hbar}TV}_T \ ,
\end{equation}
where applying $T$ to $V$ plays a role of the choice of ordering. For defining the interacting fields we use
 the Bogoliubov map
\begin{equation}
R_V(F) := S(V)^{-1} \star \left[ S(V)\cdot_T F \right ] \ .
\end{equation}
Note that for $V=0$, the above operator just reduces to $T$.

The Bogoliubov map plays the role of the intertwining M\o ller operator between the free and interacting algebras, in the sense that interacting observables are represented as formal power series in $\hbar$ and $V$ in the appropriate extension of the free algebra. This is done by means of the Bogoliubov map $F\mapsto R_V(F)$. 

We now extend the definitions above to the antifields. As for the functionals of the physical field configurations, we restrict the attention for now to the \textit{regular multi-vector fields} $\Lambda \mathscr V_{reg}$, that are, the multi-vector fields whose nth-derivatives are all smooth. 

\begin{remark}
If we think of $X \in \mathscr V_{reg}$ as a section, i.e., as 
a map from $\conf $ to $\conf_c $, $T$ acts as a differential operator and it is natural to set
\begin{equation} \label{def:T-product on antifields}
T(X) := \int_x T(X(x)) \frac{\delta}{\delta \varphi(x)} \ .
\end{equation}
\end{remark}

The non-commutative product is extended to vector fields as
\begin{equation}
X \star Y := e^{\hbar \Gamma_{\Delta_+}} (X \wedge Y) \ .
\end{equation}

The BV structures of the classical algebra get translated into structures of the quantum algebra by means of the time-ordered product. The graded algebra of antifields is transformed into $T(\mathscr V_{reg})$, with the \textit{time-ordered antibracket}
\begin{equation}
\{ X , Y \}_T := T \{ T^{-1} X , T^{-1} Y \} \ .
\end{equation}
The equations of motion are mapped, under time-ordering, into the image of the \textit{time-ordered Koszul operator}
\begin{equation}
\delta^T := T \delta T^{-1} = \{ \ \cdot \ , I_0 \}_T \ .
\end{equation}

Since $I_0$ is quadratic in the fields, and $\Delta_F$ is ($i$ times) a Green function for the equations of motion, we have a relation between the time-ordered and the classical free equations of motion 
\begin{equation} \label{eq:classical and time-ordered EOM}
\{ F, I_0 \} = \{ F, I_0 \}_T - i \hbar \bv F \ ,
\end{equation}
where the \textit{BV laplacian} is a nilpotent operator that acts on regular multi-vector fields as a divergence:
\begin{equation}
\bv X := (-1)^{1 + \abs{X}} \int_x \frac{\delta^2 X}{\delta \varphi(x) \delta \varphi^\ddagger(x)} \ .
\end{equation}
The BV laplacian combines well with the antibracket; in particular $\forall P, Q \in \Lambda \mathscr A_{\text{reg}} $ and $\forall X, Y \in T(\Lambda \mathscr A_{\text{reg}} )$, the two following formulas hold:

\begin{align}
\{ P , Q \} &= \bv(PQ) - \bv(P) Q -(-1)^{\abs{X}}P \bv{Q} \ , \ \text{and} \\
\{ X , Y \}_T &= \bv(X \cdot_T Y) - \bv(X) \cdot_T Y -(-1)^{\abs{X}}  X \cdot_T \bv{Y}  \ .
\end{align}
%



Finally, we introduce the \textit{quantum BV operator} as the deformation of the operator $\{ \cdot , I_0 \}$ under the Bogoliubov map
\begin{equation} \label{eq:interacting brst operator}
\hat s := R_V^{-1} \circ \{ \ , I_0 \} \circ R_V \ ,
\end{equation}
which, assuming the \emph{quantum master equation} (QME)
\begin{equation}\label{eq:QME}
\{ S( V(f)), I_0(f) \} = 0 \ , 
\end{equation}
can be rewritten in a more standard form \cite{FR} as
\begin{equation}
\hat s F = \{ F , I(f) \}_T - i \hbar \bv F \,,
\end{equation}
where $f\equiv 1$ on the causal completion of the support of $F$. Note that $f$ can be a multiplet of text funmctions, as discussed before.

The QME is an important condition in BV quantisation and it suffices to establish the on-shell gauge-fixing independence of the $S-$matrix and of the physical observables (the cohomology of $\hat s$). A detailed discussion of the QME and its first derivation in the context of the BV formalism in pAQFT can be found in \cite{Fredenhagen2013}.

Using the relation \eqref{eq:classical and time-ordered EOM}, it was shown in Ref. \cite{Fredenhagen2013} that Eq. \eqref{eq:QME} is equivalent to 
\begin{equation}\label{eq:QME2}
\frac{1}{2}\{L(f),L(f)\} - i\hbar \triangle L(f)=0\,,
\end{equation}
which is the more familiar form of QME. 


\subsection{Renormalization}
The construction presented above for the time-ordered products is well-defined on regular functionals only. The strategy of causal renormalization is to provide an inductive construction for the time-ordered products on the number of factors. The construction of renormalised time-ordered products has been carried out in general globally hyperbolic spacetimes for scalar and Yang-Mills theories \cite{Brunetti2009, HollandsWald2001a, HollandsWald2001b, Hollands2008} using Epstein and Glaser renormalization procedure \cite{EpsteinGlaser1973}. We refer the reader to those works for details; here, we simply recall the main steps. 

The starting point is the axiomatic definition of a family of linear maps
$T_n: \mathscr{F}^{\otimes n}_{\text{loc}} \to {\mathscr{F}}_{\mu c}[[\hbar]]$, with
\[
\supp T_n(F_1,...,F_n) \subset \bigcup \supp F_i \ ,
\]
satisfying a certain set of axioms (see e.g. \cite{Brunetti2009,HollandsWald2001b}). The most important of these axioms is the \textit{causal factorisation property}: whenever there is a Cauchy surface $\Sigma$ such that the functionals $F_1,...,F_k$ are localised in the future of $\Sigma$, and the functionals $F_k,...,F_n$ in its past, we have
\begin{equation}\label{eq:causal}
T_n(F_1,...,F_n) = T_k(F_1,...,F_k) \star T_{n-k}(F_{k+1},...,F_n) \ .
\end{equation}

We set $T_1: \mathscr F_{\text{loc}}^1 \to \mathscr F_{\mu c}[[\hbar]]$ as $T_1 := \exp{\frac{\hbar}{2} \int_{x,y} w(x,y) \frac{\delta^2}{\delta \varphi(x) \delta \varphi(y)}}$, where $w$ is the smooth part of $\Delta_+$. Up to an arbitrary length scale, $T_1$ is uniquely determined by this formula, and using $w$ ensures good covariance properties of $T_1$. Notice that, for Minkowski vacuum, $w= 0$ and $T_1$ coincides with the identity. Using the causal factorisation property, together with the other axioms, and the starting element $T_1$, it is possible to define inductively the maps $T_n$.

%

At order $n$, the map $T_n$ is constructed using the lower order maps $T_k \ , k < n$ and is uniquely fixed up to the renormalization freedom, parameterised by an $n-$linear map
\begin{equation}
Z_n : \mathscr F_{\text{loc}}^{\otimes n} \to \mathscr F_{\text{loc}}[[\hbar]] \ .
\end{equation}
$Z_n$ parametrises the freedom of adding local, locally covariant, finite counterterms in the definition of the normal-ordered observables \cite{Brunetti2009,HollandsWald2001a,HollandsWald2001b}.

The linear maps $T_n$ can now be used to define a renormalised time-ordered product. First, denote with $\mathscr F^{(0)}_{\text{loc}}$ the space of local functionals that vanish on the zero field configuration $\varphi= 0$. It is possible to prove that the pointwise multiplication map $m: S^\bullet \mathscr F^{(0)}_{\text{loc}} \to \mathscr F  $ from the symmetric tensor powers of $\mathscr F^{(0)}_{\text{loc}}$ to the space of multilocal functionals is bijective (\cite{Fredenhagen2013}, Theorem 1). Then,

using the fact that one can uniquely factorise multilocal functionals into local functionals that vanish at zero, we define a map $T_r := \oplus_n T_n\circ m^{-1}$. 

The \textit{renormalised} time-ordered product of elements of $T_r(\mathscr{BV}[[\hbar]])$ is then defined by
\begin{equation}\label{eq:renormalized-time-ordered-product}
F \cdot_{T_r} G:= T_r(T_r^{-1} F\cdot  T_r^{-1} G) \ .
\end{equation}

The renormalized time-ordered product is defined on vectors as in \eqref{def:T-product on antifields}, i.e.:
\[
T_r(X) := \int_x T_r(X(x)) \frac{\delta}{\delta \varphi(x)} \ .
\]
From the renormalized time-ordered products one can define the renormalized $S-$matrix $S_r(V) := e_{T_r}^{\frac{i}{\hbar} T V}$ and the renormalized Bogoliubov map.

The classical BV structure is then translated at the quantum level by the use of the renormalized time-ordering operator $T_r$, and we can define the renormalized time-ordered antibracket and the renormalized time-ordered Koszul map in analogy with their non-renormalized counterparts. The main problem of the extension of the BV formalism to local functionals is the BV Laplacian $\bv$, which gives raise to divergences, if not renormalised.

The relation between the $\star-$antibracket with $I_0$ and the $T_r-$antibracket with $I_0$ is given by the \textit{Master Ward Identity} \cite{Brennecke2008, Duetsch2002, Hollands2008}, which, for any local functional $\tilde V$, reads
\begin{multline}\label{eq:MWI-general}
   \{ S_r(\tilde V) \cdot_{T_r}, I_0 \} = \frac{i}{\hbar} S_r(\tilde V) \cdot_{T_r} \left [ \{I_0 + \tilde V , I_0 + \tilde V\}_{T_r} - i \hbar \An(\tilde V) \right ] \,, 
\end{multline}
and in the case of the interaction $V$ reduces to
\begin{multline} \label{eq:MWI}
\{ S_r(V) \cdot_{T_r}, I_0 \} = \frac{i}{\hbar} S_r(V) \cdot_{T_r} \left [ \{V , I_0 \}_{T_r} +\frac{1}{2}\{V,V\}_{T_r} - i \hbar \An(V) \right ]\\
=
 \frac{i}{\hbar} S_r(V) \cdot_{T_r} \left [\frac{1}{2}(\{I,I\}_{T_r}-\{I_0,I_0\}_{T_r}) - i \hbar \An(V) \right ]
\,,
\end{multline}
where $S_r$ is the renormalised S-matrix. The MWI implicitly defines the \emph{anomaly} $\An$. The anomaly can be constructed from a family of linear maps $A_n : T_r (\mathscr{BV})^{n+1} \to \mathscr A_{\text{loc}}$ as $\An = \sum_{n=0}A_n$. Each linear map $A_n$ contains $n+1$ powers of the interaction $V$, and they can be given explicitly by a recursive formula \cite{Brennecke2008, Rejzner2015}.

The \emph{renormalized BV Laplacian} $\bv_V:\mathscr F_{\text{loc}}[[\hbar]]^{\otimes n} \to \mathscr F_{\text{loc}}[[\hbar]]$ is defined in terms of the anomaly as
\[
\bv_V(X) := \frac{d}{d\lambda} \An(V+\lambda X)\big|_{\lambda=0}\,.
\]


To ensure the gauge independence of the renormalised $S-$matrix, the \textit{Quantum Master Equation} gets modified in the \textit{renormalised} QME
\[
\{S_r(V), I_0 \} = 0 \ .
\]
Equation \eqref{eq:MWI} allows us also to write down the renormalised QME as

\begin{equation}\label{eq:QMEren}
 \{V , I_0 \}_{T_r} +\frac{1}{2}\{V,V\}_{T_r} - i \hbar \An(V)=0\,,
\end{equation}
and the renormalised BV operator as
\begin{equation} \label{eq:interacting brst operator ren}
\hat s := R_V^{-1} \circ \{ \ , I_0 \} \circ R_V= \{ F , I \}_{T_r} - i \hbar \bv_V F \ .
\end{equation}
Using the QME for $I_0$ we can rewrite \eqref{eq:QMEren} in a simpler form that only involvess $I$ and $I_0$:
\begin{equation} \label{eq:rQME0}
\frac{1}{2} \{ I , I \}_{T_r}-\frac{1}{2} \{ I_0 , I_0 \}_{T_r} = i \hbar \An(V) \ . 
\end{equation}
Assuming that also $I_0$ satisfies the classical master equation, we obtain:
\begin{equation} \label{eq:rQME}
\frac{1}{2} \{ I , I \}_{T_r} = i \hbar \An(V) \ .
\end{equation}

In the following, to avoid heavy notation we will denote the renormalized time-ordered products as the non-renormalized ones, $\cdot_T$; whenever we are dealing with local functionals, it will always implicitly assumed that we are using the renormalized quantities. The only cases where a functional is both regular and local are constant and linear functionals. For those, the time-ordering map $T$ is the identity, both for the renormalized and the non-renormalized prescription. The BV Laplacian is trivial for these cases.


\section{Generating functionals} \label{sec:generating functionals}

The starting point to formulate the Euclidean Functional Renormalization Group is usually a regularization of the generating functional \cite{Wetterich1992}, formally defined by a path integral over the field configuration space. One defines
\begin{equation}
\mathcal Z_k = \int \mathcal D \varphi e^{-\frac{1}{\hbar}(I + Q_k)} \ ,
\end{equation}
where $I$ is the gauge-fixed action and $Q_k$ is an IR regulator term, quadratic in the fields, that is usually taken as local in momentum space but non-local in position. Such a regulator acts as momentum-dependent mass term, and in the flow equations regularises both the IR and UV divergences.

It is difficult, however, to formulate a rigorous definition of the path integral on Lorentzian spacetimes. Moreover, the interpretation of $\mathcal Z_k$ as the generating functional of correlation functions seems more problematic for curved spacetimes or in generic states. For these reasons, in \cite{DDPR2022} an approach based on pAQFT has been proposed. Here we follow the same approach, recalling the main definitions. First, we define the source term as a smeared field with smooth functions $j$,
\begin{equation}\label{def:sources}
J(\varphi) := \int_x j(x) \varphi(x) \ .
\end{equation}
Whenever the field configurations $\varphi$ are integrated, we are always also assuming a summation over Lorentz as well as internal (as the gauge group component) indices. In addition to linear sources, in order to understand the renormalization of symmetries we will also need to include effective sources in the generating functional,
\begin{equation} \label{def:BRST sources}
\Sigma := I_{af}(\varphi, \varphi^\ddagger = \sigma) = \eval{\alpha_{\sigma \varphi}(I_{af})}_{\varphi^\ddagger = 0} \,,
\end{equation}
where $\alpha$ is defined in \eqref{eq:alpha def} and $\sigma$ is compactly supported in spacetime.

In the case of theories that are linear in the antifields, as Yang-Mills and gravity, $\Sigma$ can be written as
\[
\Sigma = \int_x \sigma(x) \{ \varphi , I_{af} \} = \int_x \sigma(x) \frac{\delta V}{\delta \varphi^\ddagger(x)} \ .
\]
Moreover, it holds that
\[
V + \Sigma = \alpha_{\sigma \varphi}(V) = V(\varphi, \varphi^\ddagger + \sigma) \ ,
\]
so that $\Sigma$ behaves as a gauge-fixed interaction $V$, in which the antifields $\varphi^\ddagger$ are substituted by effective BRST sources $\sigma$.
\begin{lemma}\label{lemma:sigma}
In the unrenormalised theory (or, equivalently, in absence of anomalies), $\Sigma$ is invariant under $\{ \ \cdot \ ,V\}_T$.
\end{lemma}
\begin{proof}
We prove this by direct computation. Since $\Sigma = \int_x \sigma(x) s \varphi(x)$, we have
\begin{equation*}
\{ V, \Sigma \}_T = \{ I , \Sigma \}_T = \hat s \Sigma + i \hbar \bv \Sigma =
\int_x \sigma(x) \hat s^2 \varphi(x) - i \hbar \bv(\Sigma) = 0 \ .
\end{equation*}
We used the definition of $\hat s$ in the first line, and the nilpotency of $\hat s$ together with the fact that $\Sigma$ is independent on the antifields to conclude.

\end{proof}

Finally, we need to insert a regulator term in the generating functional. Although the regulator term is usually taken to be local in momentum, this is unsuitable for theories on curved space, where a Fourier transform might not even be possible without spoiling the covariance of the theory. Moreover, a regulator non-local in position might spoil the unitarity of the $S-$matrix, and introduce artificial poles in the propagator. Therefore, we choose to introduce a local regulator,
\begin{equation}
Q_k := - \frac{1}{2} \int_x q_k(x)  \varphi^2(x) \ ,
\end{equation}
where $q_k$ is compactly supported and smooth as a spacetime function.

The use of a local regulator has been discussed in the literature, and has found recent applications in e.g. \cite{Fehre2021}. A local regulator does not implement a Wilsonian renormalization, in the sense that it does not imply a successive integration of momentum shells, but rather it implements a flow in the space of massive theories. Moreover, being simply a mass term, the use of a local regulator needs an additional UV renormalization in the flow equation. In our formalism, such renormalization is implemented by a Hadamard subtraction of divergences \cite{DDPR2022}. For different approaches to Lorentzian fRG equations, one can also see \cite{Banerjee2022, Manrique2011}.

\begin{example}[Yang-Mills theories]
In the case of Yang-Mills theories, we can write explicitly the new contributions $\Sigma$ and $Q_k$. $\Sigma$ by definition is
\[
\Sigma = I_{af}(\varphi^\ddagger = \sigma) = \int_x D_\A c \sigma_\A - \frac{i\lambda_{YM}}{2} [c , c] \sigma_c \ ,
\]
while $Q_k$ is
\[
Q_k = - \frac{1}{2} \int_x q_k^A(x) T \abs{A}^2(x) - q_k^c(x) T( \bar c c - c \bar c) \ .
\]
\end{example}

Having introduced all ingredients, we can now define the \textit{regularised generating functional},
	\begin{equation} \label{def:Z_k}
	Z_k(\varphi, \varphi^\ddagger ; j, \sigma ) := S(V)^{-1} \star S(V + \Sigma +Q_k + J) = R_V (S( Q_k + J + \Sigma)) \ .
	\end{equation}

The regularised generating functional $Z_k$, and, similarly, $W_k$ and the effective average action, defined below, are functionals mapping an element of the dual of the extended field configuration space and valued in the formal power series of the free algebra, $Z_k: \overline{\mathscr E}'  \to \mathscr{BV}[[\lambda, \hbar]]$. This algebraic definition differs from the previous definition of generating functionals given in the case of scalar theories \cite{DDPR2022}, where we evaluated $Z_k$ in a concrete state, so it was a map from $\overline{\mathscr E}'$ to $\mathbb{C}$. Now $Z_k$ is a map between two locally convex topological vector spaces, so one needs to invoke a more general concept of differentiability. We will again use the definition of Bastiani (Definition~\ref{df:Bastiani}, but with $\mathscr{BV}[[\lambda, \hbar]]$ as a  target). The topology on $\mathscr{BV}[[\lambda, \hbar]]$ is induced from having the H\"ormander topology at each fixed order of $\lambda$ and $\hbar$. Differentiability of $Z_k$ with respect to $j$ follows from the fact, can be demonstrated explicitly, since
    \[
    Z_k(\varphi, \varphi^\ddagger ; j+t\psi, \sigma )=S(V)^{-1} \star \left(S\left(V + \Sigma +Q_k + J\right)\cdot_T S\left(\int_x \psi(x)\ph(x)\right)\right)\,.
    \]
and hence
\begin{multline*}
\frac{1}{t}( Z_k(\varphi, \varphi^\ddagger ; j+t\psi, \sigma )- Z_k(\varphi, \varphi^\ddagger ; j, \sigma ))
\\=S(V)^{-1} \star \left(S(V + \Sigma +Q_k + J)\cdot_T \frac{1}{t}\left(S\left(\int_x \psi(x)\ph(x)\right)-1\right)\right)\,.
\end{multline*}
The limit $t\rightarrow 0$ exists and is given by
\[
S(V)^{-1} \star \left(S(V + \Sigma +Q_k + J)\cdot_T \frac{1}{t}\left(\int_x \psi(x)\ph(x)\right)\right)\,,
\]
which is an element of $\mathscr{BV}[[\lambda, \hbar]]$, as required. Note that this expression can be seen as a vector valued distribution (in this case valued in $\mathscr{BV}[[\lambda, \hbar]]$), in the sense of Schwarz \cite{Schwarz1,Schwarz2}.

One can iterate this process to obtain higher derivatives as well. The key point is that one always ends up differentiating an exponential.

    The new definition highlights the algebraic properties of the generating functionals, stressing that the structural form of the RG flow equation does not depend on the choice of a state. The physical interpretation of $Z_k$ as the generating functional of interacting correlation functions for a given state is recovered, once $Z_k$ is evaluated on a fixed Hadamard state for the free theory. The state evaluation of the algebraic definition coincides with the previous definition for the generating functional, given in Ref. \cite{DDPR2022}. 

    Finally, notice that, since the state for the free theory does not depend on the RG scale $k$, the state evaluation and the scale derivative with respect to $k$, giving raise to the RG flow equation, commute. 
 
The defining property of $Z_k$ is that its functional derivatives, at vanishing sources and for vanishing regulator, give the time-ordered correlation functions
\begin{equation}
\eval{\frac{\delta^n Z_k}{\delta j(x_1)... \delta j(x_n)}}_{j = \sigma = k = 0} = R_V \left( \varphi(x_1) \cdot_T... \cdot_T \varphi(x_n) \right) \ .
\end{equation}

To justify this interpretation, we recall that, on regular functionals, we can introduce an interacting star product as \cite{Fredenhagen2011}
\begin{equation}
F \star_V G = R_V^{-1} \left ( R_V(F) \star R_V(G) \right ) \ ,
\end{equation}
as well as an interacting state built from a state $\omega$ of the free theory as $\omega_V = \omega \circ R_V$. The correlation functions of interacting observables in the interacting state are given by
\begin{equation}
\omega_V(\varphi(x_1) \star_V ... \star_V \varphi(x_n) ) = \omega ( R_V(\varphi(x_1) \star ... \star R_V(\varphi(x_n)) ) \ .
\end{equation}

Finally, we recall that, if the support of an observable $F$ does not intersect the past of the support of $G$ ($F \gtrsim G$), we have \cite{Drago2015, Lindner2013}
\begin{equation}
R_V(F \cdot_T G) = R_V(F) \star R_V(G) \ .
\end{equation}
Therefore, the functional derivatives of $Z_k$ are the time-ordering of the interacting correlation functions in the interacting state. The regulator term introduces a $k-$dependent mass in the propagators via the Principle of Perturbative Agreement \cite{DDPR2022, Drago2015, HollandsWald2004, Zahn2015}.

From the definition of $Z_k$, we follow the usual definitions of the other generating functionals \cite{Weinberg-vol2}: first, we define the \textit{regularised generator of connected correlation functions} $W_k$ as
	\begin{equation} \label{def:W_k}
	e^{\frac{i}{\hbar} W_k(\varphi, \varphi^\ddagger; j, \sigma)} := Z_k \ .
	\end{equation}

We introduce the \textit{mean value operator} as a bit of useful notation,
	\begin{equation} \label{def:mean value}
	\langle F \rangle = e^{-\frac{i}{\hbar}W_k} (S(V)^{-1} \star \left [ S(V+Q_k + J + \Sigma) \cdot_T F \right ]) \ .
	\end{equation}

The mean value operator maps an element of the free algebra in the formal power series of the interacting algebra, with coefficients that depend on the sources and the mean antifields $j, \ \sigma$ as well as on the RG scale $k$. Notice that the mean value operator does not involve state evaluation. The mean value operator takes only one input, and no confusion should arise with the standard pairing between distributions $\langle \ \cdot \ , \ \cdot \ \rangle$.
 
Functional derivatives of $W_k$ give the connected correlation functions. Its first derivative defines the \textit{effective field} $\phi$:
\begin{equation}\label{def:phi}
\phi_j := \frac{\delta W_k}{\delta j} \ ,
\end{equation}
while the second derivative gives the \textit{connected interacting propagator}
\begin{equation} \label{second derivative W_k}
   - i \hbar \frac{\delta^2 W_k}{\delta j(x) \delta j(y)} = \langle \varphi(x) \cdot_T \varphi(y) \rangle - \phi(x) \phi(y) \ .
\end{equation}

The relation between $\phi$ and $j$ can be perturbatively inverted \cite{DDPR2022, Ziebell2021}, giving $j$ as a function of $\phi$; we define $\tilde \Gamma_k$ as the Legendre transform of $W_k$,
	\begin{equation} \label{def:tilde Gamma}
	\tilde{\Gamma}_k = W_k(j_\phi) - J_\phi(\phi) \ .
	\end{equation}
Finally, we define the \textit{effective average action} as a modified Legendre transform of $W_k$:
    \begin{equation} \label{def:Gamma_k}
	\Gamma_k(\varphi, \varphi^\ddagger ; \phi, \sigma) = W_k(j_\phi) - J_\phi(\phi) - Q_k(\phi) \ .
	\end{equation}
$\Gamma_k$ will be the main object of our discussion, as it is the central tool in the functional Renormalization Group \cite{Dupuis2021}. Its first and second derivatives satisfy two important properties, which are easily proved by direct inspection. The first derivative gives
\begin{equation} \label{eq:qeom}
\frac{\delta }{\delta \phi} ( \Gamma_k + Q_k(\phi) ) = - j_\phi \ ,
\end{equation}
while the second derivative is the inverse of the full propagator $\frac{\delta^2W_k}{\delta j^2}$:
\begin{equation} \label{eq: relation Gamma and its inverse G}
\int_z \frac{\delta^2 (\Gamma_k + Q_k)}{\delta \phi(x) \phi(z)} \frac{\delta^2 W_k}{\delta j(z) \delta j(y)} = - \delta(x,y) \ .
\end{equation}

The above equations can be considered as the "quantum" version of the equations of motion and of the definition of the propagator as the Green's function for the wave operator, respectively. 

From \eqref{eq:qeom}, it is possible to derive \cite{DDPR2022}
\begin{equation}
    \frac{\delta \Gamma_k}{\delta \phi} = \frac{\delta I_0}{\delta \phi} + \langle V^{(1)} + \Sigma^{(1)} \rangle \ .
\end{equation}
In the classical limit $\hbar \to 0$, the non-commutative and time-ordered products reduce to the point-wise product. The Bogoliubov map in turn reduces to the classical M\o ller operator, which acts on functionals of the field by pull-back, so that $\langle F(\varphi) \rangle = F(\phi) + \mathcal O(\hbar)$. Therefore we obtain, modulo an irrelevant constant,
\begin{equation} \label{eq:classical limit Gamma}
    \Gamma_k(\varphi^\ddagger; \phi, \sigma) = I(\phi, \varphi^\ddagger) + \Sigma(\phi, \sigma) + \mathcal O(\hbar) \ .
\end{equation}

Compared with our previous work \cite{DDPR2022}, here we do not need to state evaluate the generating functionals. We see that our definitions, as well as the flow equations, hold at the pure algebraic level. The generating functionals $Z_k$, $W_k$, and $\Gamma_k$ are interpreted as functionals of $(j,\sigma)$ or of $(\phi, \sigma)$, parametrised by the field and antifield configurations $(\varphi, \varphi^\ddagger)$. In order to extract physical predictions, one can evaluate the functionals in some state $\omega$. The state evaluation also provides the connection between the definitions given here and those in our previous paper.

However, the state dependence still enters the flow equation through the choice of a Green's function for the quantum wave operator $\frac{\delta^2}{\delta \phi^2} (\Gamma_k + Q_k)$. In fact, the main difference from the euclidean case is that, in globally hyperbolic spacetimes, the Green's function for a wave-type operator is not unique.

The Green's function $G_k$ for $\frac{\delta^2}{\delta \phi^2} (\Gamma_k + Q_k)$ can be constructed as a Dyson-type series from the free Feynman propagator $\Delta_F$.
In turn, the Feynman propagator depends on the arbitrary choice of the symmetric part of the distribution $\Delta_+$. Such a choice can be connected to the choice of the state two-point function. Indeed, for the state given by $\omega_0(F):= F(0)$, we have
\[
\omega_0(\varphi(x) \star \varphi(y) ) = \Delta_+(x,y) \ .
\]
	
\section{Flow equations} \label{sec:flow-eqs}

The flow equations for the generating functionals, and in particular for the effective average action, can be directly computed by taking the $k-$derivative of each definition, just as in the scalar case \cite{DDPR2022}. We briefly recall the steps, starting from the flow for $Z_k$, which by direct inspection is
\begin{equation}
\partial_k Z_k(\varphi, \varphi^\ddagger ; j, \sigma) = \frac{i}{\hbar} S(V)^{-1} \star \left [ S(V+Q_k+J + \Sigma) \cdot_T \partial_k Q_k \right ] \ .
\end{equation}
From the definition of $W_k$ \eqref{def:W_k}, we get
\begin{equation}
\partial_k W_k = e^{- \frac{i}{\hbar} W_k} (S(V)^{-1} \star \left [ S(V+Q_k+J + \Sigma) \cdot_T \partial_k Q_k \right ]) \ .
\end{equation}
The above can be rewritten in terms of the second derivative of $W_k$, by first expressing the r.h.s as
\begin{equation}
\partial_k W_k = - \frac{1}{2} \int_x \partial_k q_k(x) \langle T \varphi^2 \rangle
\end{equation}

Since the second derivative of $W_k$ gives the connected part of the full propagator \eqref{second derivative W_k}, we can rewrite the flow equation for $W_k$ in the form of a local \textit{Polchinski equation},
\begin{equation}
\partial_k W_k = \frac{1}{2} \int_x \partial_k q_k \left ( i \hbar : W^{(2)}_k :_{\tilde H_F} - W_k^{(1)}(x) W_k^{(1)}(x) \right ) \ ,
\end{equation}
where the normal-ordering is defined as
\begin{equation}
: W^{(2)}_k:_{\tilde H_F} = \lim_{y \to x} W^{(2)}_k(x,y) - \tilde H_F (x,y) \ .
\end{equation}

The counterterm $\tilde H_F$ is implicitly defined by the relation
\begin{equation}
\lim_{y \to x} \langle \varphi(x) \cdot_T \varphi(y) \rangle - \tilde H_F(x,y) 
:= \langle \lim_{y \to x} \varphi(x) \cdot_T \varphi(y) - H_F(x,y) \rangle \ ,
\end{equation}
where $H_F$ is the Feynman Hadamard parametrix defined in Remark \ref{rmk:Hadamard-Feynman}. $\tilde H_F$ can be computed by a Hadamard expansion of the relation \eqref{eq: relation Gamma and its inverse G}.

Finally, since
\begin{equation}
\partial_k \Gamma_k = \partial_k W_k - \partial_k Q_k(\phi) = \partial_k W_k + \frac{1}{2} \int_x \partial_k q_k \phi^2(x) \ ,
\end{equation}
and recalling that $\phi = W_k^{(1)}$ by definition (Eq. \eqref{def:phi}),
we arrive at the RG flow equation for gauge theories, which takes the form of a local \textit{Wetterich equation} (also known as \textit{functional Callan-Symanzik equation} \cite{Alexandre2001}), generalised to the algebraic level and for Lorentzian spacetimes:
\begin{equation} \label{eq:wetterich}
\partial_k \Gamma_k = \frac{i \hbar}{2} \int_x \partial_k q_k(x) G_k(x,x) \ ,
\end{equation}
together with the condition
\begin{equation} 
G_k(x,x) = : W^{(2)}_k(x,x) :_{\tilde H_F} = - : \left ( \Gamma^{(2)}_k - q_k \right )^{-1} :_{\tilde H_F}\,.
\end{equation}

\section{Symmetries} \label{sec:symmetries}
\subsection{Modified Slavnov-Taylor identities}
\label{sec:mSTI}

The QME governs gauge dependence of the action and its renormalization \cite{Fredenhagen2013}. Now, we would like to understand what are the symmetry constraints to be imposed on the effective average action, to be consistent with BRST invariance. The identity satisfied by the effective average action is usually called in the fRG context \textit{modified Slavnov-Taylor identity}. In the pQFT approach based on a path integral formulation, the Slavnov-Taylor identity is derived from the gauge invariance and parametrization invariance of the generating functional $Z$. In the absence of a regulator term (or more precisely, assuming an implicit BRST invariant regularization), the Slavnov-Taylor identity assumes the form of the \textit{Zinn-Justin equation} (\cite{Zinn-Justin1975, Zinn-Justin2002} )
\begin{equation}
\int_x \frac{\delta \Gamma_0}{\delta \phi(x)} \frac{\delta \Gamma_0}{\delta \sigma(x)} = 0 \ .
\end{equation}

The Zinn-Justin equation is usually interpreted as a symmetry constraint for $\Gamma_0$ and it plays a crucial role in the perturbative renormalization of gauge theories. In the presence of a non-invariant regulator term, the Zinn-Justin equation is modified by a symmetry breaking term that, in the context of fRG, is written as \cite{Ellwanger1994}
\begin{equation} \label{eq:standard mSTI}
\int_x \frac{\delta \Gamma_k}{\delta \phi(x)} \frac{\delta \Gamma_k}{\delta \sigma(x)} = i \hbar \int_{x,y} q_k(x) G_k(x,y) \frac{\delta^2 \Gamma_k}{\delta \sigma(x) \delta \phi(y)} \ .
\end{equation}
A detailed discussion of the mSTI can be found e.g. in \cite{Pawlowski2005}.

Since the regulator term is usually non-local, $q_k = q_k(x,y)$, the r.h.s of the above equation is usually written as a trace over a loop contribution. The equation above is called \textit{modified Slavnov-Taylor identity} and is interpreted as the breaking of the symmetries of the effective action by the regulator dependent term. In the $k \to 0$ limit, the mSTI reduces to the Zinn-Justin equation. Therefore, if the effective average action satisfies the mSTI, this ensures that the effective action $\Gamma_0 = \lim_{k \to 0} \Gamma_k$ satisfies the usual Zinn-Justin equation, and in this sense the symmetries of the effective action are restored. 

Since the mSTI are derived by the same functional as the flow equations, an exact solution to the Wetterich equation also satisfies the mSTI, governing how the symmetries are broken and restored along the flow. 

Although conceptually satisfying, for practical computations one needs to rely on some kind of approximation to solve the flow equation, usually in the form of a truncation in the parameter space. In this case the mSTI becomes a non-trivial constraint on the truncation, and solving the flow equation with an approximated effective average action that also satisfies the mSTI can become a difficult task. In fact, a standard strategy \cite{Reuter1996} is to find a suitable truncation scheme for the flow equation, which in general does not satisfy the mSTI, and measure the quality of the approximation by the order of magnitude of the mSTI breaking term.

The loop structure of the regulator term on the r.h.s lets one argue that it is a higher order contribution, and, in a first approximation, it can be neglected \cite{Reuter1996, ReuterWetterich1994}. However, since the regulator-dependent term is proportional to a second-order derivative of the effective average action, one cannot use cohomological methods \cite{Barnich2000} to discuss the renormalization of the effective action. For these reasons, many alternatives have been pursued in the literature, in particular in constructing manifestly gauge invariant flows. However, these attempts usually come at the price of rather involved computations. 

In the following section, we will modify the generating functional $Z_k$ in order to obtain the standard Slavnov-Taylor identities, introducing an additional field $\eta$.
Now, we briefly show how to derive the mSTI in our context, following the same steps as in the path integral approach.  We derive the results in the unrenormalised case, for simplicity of comparison with the literature.

We start by computing the action of the free  Koszul-Tate differential on $Z_k$
\begin{equation}
\delta_0 Z_k = S(V)^{-1} \star \left [ \delta_0 S(V+ J + Q_k + \Sigma) \right] \ .
\end{equation}
Now, we use the AMWI \eqref{eq:MWI} (where, in the unrenormalised case, the anomaly simply reduces to the BV laplacian) to conclude that
\begin{multline*}
 \delta_0 S(V+ J + Q_k + \Sigma)\\=\frac{i}{\hbar} S(V+ J + Q_k + \Sigma)\cdot_T (\frac{1}{2}\{I+ J + Q_k + \Sigma,I+ J + Q_k + \Sigma\}_T-i\hbar A(V+ J + Q_k + \Sigma))\\
 =\frac{i}{\hbar} S(V+ J + Q_k + \Sigma)\cdot_T (\{J + Q_k + \Sigma,V\}_T+\frac{1}{2}\{I,I\}-i\hbar A(V+ J + Q_k + \Sigma)) \ .
\end{multline*}
We observe that $\{Q_k+ J + \Sigma, I_0 \} = 0$, since $J+Q_k + \Sigma$ does not contain antifields and the QME \eqref{eq:QME} implies that  $\frac{1}{2}\{I,I\}=i\hbar A(V)$, so we can write
\begin{equation} \label{eq:intermediate s_0Z_k}
\delta_0 Z_k =
\frac{i}{\hbar} S(V)^{-1} \star \left(S(V+Q_k + J + \Sigma )\cdot_T\left(\{Q_k + J + \Sigma , V \}_T - i \hbar A_V(J+Q_k+\Sigma)\right)\right)\,,
\end{equation}
where we denoted $A_V(J+Q_k+\Sigma)\equiv A(V+J+Q_k+\Sigma)-A(V)$.

Thanks to Lemma \ref{lemma:sigma}, $\Sigma$ does not contribute to the antibracket with $V$ in the above equation, so we obtain
\begin{equation}
\delta_0 Z_k = \frac{i}{\hbar} S(V)^{-1} \star \left [ S(V+Q_k + J + \Sigma) \cdot_T (\{ Q_k + J , V \}_T - i \hbar  A_V(J+Q_k+\Sigma))\right ]
\end{equation}

Now, the mean antifield term $\Sigma$ is introduced in such a way that
\begin{align*}
\int_x  S(V+Q_k + J + \Sigma) \cdot_T \frac{\delta V }{\delta \varphi^\ddagger(x)} &= \int_x  S(V+Q_k + J + \Sigma) \cdot_T \frac{\delta \Sigma }{\delta \sigma(x)} \\
&= - i \hbar \int_x \frac{\delta }{\delta \sigma(x)} S(V+Q_k + J + \Sigma) \ .
\end{align*}
Moreover, $J$ is linear in the antifields, so that
\begin{align*}
S(V+Q_k + J + \Sigma) \cdot_T \{ J , V \}_T &= \int_x j(x) S(V+Q_k + J + \Sigma) \cdot_T \frac{\delta V }{\delta \varphi^\ddagger(x)} \\
&= - i \hbar\int_x j(x)\frac{\delta }{\delta \sigma(x)} S(V+Q_k + J + \Sigma) \ .
\end{align*}
Finally, since $Q_k$ does not depend on $\varphi^\ddagger$ and $ \sigma$, we have
\begin{align*}
S(V+Q_k + J + \Sigma) \cdot_T \{ Q_k , V \}_T &= \int_x  S(V+Q_k + J) \cdot_T \frac{\delta V }{\delta \varphi^\ddagger(x)} \cdot_T \frac{\delta Q_k }{\delta \varphi(x)}\\
&= - i \hbar \int_x \frac{\delta}{\delta \sigma(x)} \left[ S(V+ Q_k + J + \Sigma) \cdot_T \frac{\delta Q_k}{\delta \varphi(x)} \right] \ .
\end{align*}
Using these results, we can write the action of $\delta_0$ on $Z_k$ as
\begin{multline}
\delta_0 Z_k =  \int_x j(x) \frac{\delta Z_k}{\delta \sigma(x)}\\ + \frac{\delta}{\delta \sigma(x)} \left [S(V)^{-1} \star \left ( S(V+Q_k + J + \Sigma ) \cdot_T \left(\frac{\delta Q_k}{\delta \varphi(x)}\right) \right )\right ] \\
- i \hbar  S(V)^{-1} \star \left ( S(V+Q_k + J + \Sigma ) \cdot_T A_V(J+Q_k+\Sigma) \right )  \ .
\end{multline}

The first term on the right-hand side is the standard term appearing in the Ward identity for the gauge symmetry of the generating functional $Z(j)$, the second term is the contribution from the regulator and the third term is the anomaly. 


Substituting first the definition of $W_k$ \eqref{def:W_k}, then of the effective average action \eqref{def:Gamma_k}, and recalling the identities
\begin{equation}
j - q_k \frac{\delta W_k}{\delta j_\phi} = - \frac{\delta \Gamma_k}{\delta \phi} \quad , \quad \frac{\delta^2 W_k}{\delta \sigma(x) \delta j(y)} = \int_z G_k(x,z) \frac{\delta^2 \Gamma_k}{\delta \phi(z) \delta \sigma(x)} \ ,
\end{equation}

we get the modified Slavnov-Taylor identity:
\begin{equation} \label{eq:mSTI}
\delta_0 \Gamma_k + \int_x \frac{\delta \Gamma_k}{\delta \phi(x)} \frac{\delta \Gamma_k}{\delta \sigma(x)} = i \hbar \int_{x,y}q_k(x) G_k(x,y) \frac{\delta^2 \Gamma_k}{\delta \phi(x) \delta \sigma(y)} + \langle A_V(J+Q_k + \Sigma) \rangle \ .
\end{equation}
The two differences from the standard form \eqref{eq:standard mSTI} are the presence of the Koszul variation of $\Gamma_k$, which arises because here we are considering off-shell functionals and the presence of the anomaly term. Whenever evaluated on a faithful state and in the absence of anomalies, the mSTI \eqref{eq:mSTI} reduces to the one known in the literature \eqref{eq:standard mSTI}.

\subsection{Effective master equation for the effective average action}

As discussed, the presence in the mSTI \eqref{eq:mSTI} of a second-order derivative of the effective average action prevents one from using cohomological methods to determine the structural form of $\Gamma_k$. We now derive a different symmetry identity, closely related to the mSTI, which lets us discuss the regulator term in cohomology. The idea is an adaptation of the treatment given in \cite{Zinn-Justin2002} of symmetries broken by quadratic terms. The key insight is in recognising that the regulator term is "half of a contractible pair": i.e., we can enlarge the configuration space by adding a source for $\{ V , \varphi^2(x) \}$ in the generating functional $Z_k$. We first discuss the case of $V$ linear in the antifields; we will comment later on how to proceed for a theory with $V$ at order $l$ in the antifields.
\begin{definition} \label{def:H}
We extend the space of fields $\varphi$ with a new compactly supported field $\eta$, and we define
\begin{equation} \label{eq:def H}
H := \frac{1}{2} \int_x \eta(x) \{ V , \varphi^2(x) \} \ ,
\end{equation}
together with an extended BV differential
\[
\mBV := \hat s + \int_x q_k(x) \frac{\delta}{\delta \eta} - \int_x j(x) \frac{\delta}{\delta \sigma(x)} \ . 
\]
\end{definition}

The variation of $\eta$ under $\mBV$ given by
\begin{equation}
\mBV \eta(x) = q_k(x) \ .
\end{equation}

Since $\mBV q_k = 0$, the pair $(\eta, q_k)$ forms a contractible pair in the cohomology of the extended BV differential $\mBV$.

The pair $(\eta, q_k)$ can be understood as an enlargement of the non-minimal sector. In fact, defining
\begin{equation}\label{eq:Xi}
\Xi := - \frac{1}{2} \int_x \eta(x) \varphi^2(x) \ ,
\end{equation}
a short computation shows that
\begin{equation}\label{eq:Xi-Q-H}
\mBV \Xi = H + Q_k \ .
\end{equation}

In the same way, $j$ and $\sigma$ form a contractible pair, and the source sector $J + \Sigma$ can be rewritten as
\begin{equation}\label{eq:source-bv-exact}
J+\Sigma = s_k\int_x \sigma(x) \varphi(x) \ .
\end{equation}

It follows that the source and regulator terms are $s_k-$exact:
\begin{equation}\label{eq:source-regulator-exact}
    J + \Sigma + Q_k + H = s_k \int_x \left [ \sigma(x) \varphi(x) - \frac{1}{2} \eta(x) \varphi^2(x) \right ]\ .
\end{equation}

\begin{example}[Yang-Mills theories]
In the case of Yang-Mills theories, $\Xi$ is
\[
\Xi = - \frac{1}{2} \int_x \eta ( \abs {\A}^2 + \bar c c ) \ , 
\]
and so $H$ is
\[
H = \int_x \eta(x) \left [  \A D_\A c + \frac{i}{2} ( b c + \frac{ \lambda_{YM}}{2} c [c,c] ) \right ] \ .
\]
\end{example}

From its action on the fields, one can check that $\mBV^2=0$, so $\mBV$ is a differential. 
In fact, first notice that $\mBV$ is defined as the sum of three, mutually commuting terms, since they each act on different fields: $[\hat s, \int_x q_k(x) \frac{\delta}{\delta \eta} + \int_x j(x) \frac{\delta}{\delta \sigma(x)}]= [\int_x q_k(x) \frac{\delta}{\delta \eta},\int_x j(x) \frac{\delta}{\delta \sigma(x)} ] = 0$. ${\hat s}^2=0$ thanks to rQME. Moreover, $\int_x q_k \frac{\delta}{\delta \eta}$ is nilpotent, since $\eta$ and $q_k$ are a contractible pair, so that the action of $\int_x q_k \frac{\delta}{\delta \eta}$ on any functional of the field $\eta$ (notice that any functional can be at most linear in $\eta$, since it is a fermionic field) produces a term which is proportional to $q_k$, and $\mBV q_k = 0$ by definition. Finally, $\int_x j(x) \frac{\delta}{\delta \sigma(x)}$ is nilpotent because $j$ and $\sigma$ have opposite parity: if $\sigma$ is fermionic, than any functional $F$ can be at most linear in $\sigma$, and $\int_{x,y} j(x) j(y) \frac{\delta^2 F}{\delta \sigma(x)\sigma(y)}=0$. On the other hand, if $\sigma$ is bosonic, $j(x)$ is fermionic, and $j^2 = 0$.

Thanks to the new term $H$, the BV symmetry of the original action $I$ can now be extended to a larger symmetry, encoded in the operator $\mBV$, for $I_0 + V + J + \Sigma + Q_k + H$:
\begin{equation} \label{eq:classical extended BV invariance}
\mBV (I_0 + V + J + \Sigma + Q_k + H) = 0 \ . 
\end{equation}
The action of $\mBV$ on $I_0 + V$ is simply the action of the BV differential $\hat s$, and as such it vanishes because of the QME. 
Moreover
$$\mBV(J+\Sigma + Q_k + H) = \mBV^2 \left ( \int_x \sigma(x) \varphi(x) + \Xi \right ) = 0$$
because $\mBV$ is a differential.

The inclusion of $H$, and the replacement of the BV differential with a scale-dependent operator, extends the BV invariance of the original action to a larger symmetry for the classical action $I + J + \Sigma + Q_k + H$. Now, we want to discuss the consequences of the extended classical symmetry \eqref{eq:classical extended BV invariance} to the quantum correlation functions, by analysing its consequences on the effective average action.

The additional terms $\Sigma + H$ can also be understood as a generalised gauge-fixing term: in fact, introducing the functional
\[
\Theta := \int_x \sigma(x) \varphi(x) + \frac{1}{2} \eta(x) \varphi^2(x) \,,
\]
we obtain
\[
\alpha_\Theta(V) = V\left( \varphi^\ddagger + \frac{\delta \Theta}{\delta \varphi} \right) = V \left (\varphi^\ddagger + \sigma + \eta \varphi \right ) \ .
\]

Since the rQME states the invariance of the action $I$ under gauge transformations, it holds for any gauge fixed interaction functional; in particular, it holds that \cite{FR}
\begin{equation} \label{eq:gauge-fixed-QME}
\frac{1}{2} \{ \alpha_\Theta(V) + I_0 , \alpha_\Theta(V) + I_0 \}_T - i \hbar \An( \alpha_\Theta(V) ) = 0
\ .
\end{equation}

For Yang-Mills-type theories, the rQME with generalised gauge-fixing reduces to
\[
\frac{1}{2} \{ I + \Sigma + H, I + \Sigma + H \}_T - i \hbar \An( V + \Sigma + H ) = 0
\ ,
\]
and subtracting the rQME itself, Eq. \eqref{eq:rQME}, we get
\[
 \{I, \Sigma + H \}_T + \frac{1}{2}\{ \Sigma + H , \Sigma + H \}_T - i \hbar \left ( A(V + \Sigma + H) - A(V) \right ) = 0 \ .
\]
Since $\Sigma + H$ and $I_0$ do not contain antifields, the above equation reduces to
\begin{equation}\label{eq:sigma-H-simplification}
    \{V, \Sigma + H \}_T  - i \hbar \left ( A(V + \Sigma + H) - A(V) \right ) = 0 \ .
\end{equation}

From $\alpha_\Theta(V)$ we define the new, $\sigma-$ and $\eta-$dependent regularising generating functional $Z_k$ as
\begin{equation}
Z_k := S(V)^{-1} \star S(\alpha_\Theta (V) + Q_k + J) \ .
\end{equation}
Notice that the Definition of the $Z_k$ functional given above is valid for both Yang-Mills-type theories, as Yang-Mills and gravity, which are linear in the antifields, and for more general theories with additional, $\varphi^\ddagger-$non-linear contributions.

From this re-definition of $Z_k$, we define $W_k$ and $\Gamma_k$, as in \eqref{def:W_k} and \eqref{def:Gamma_k}. Since $\eta$ is a classical field, we have $\frac{\delta W_k}{\delta j_\eta} = \eta$ and $\frac{\delta \tilde \Gamma_k}{\delta \eta} = - j_\eta$.

In the case of Yang-Mills-type theories, linear in the antifields, $\alpha_\Theta(V)$ reduces to $V + \Sigma + H$, and the regularising generating functional becomes
\begin{equation} \label{def:extended-Zk}
Z_k(\varphi, \varphi^\ddagger ; \sigma, \eta, j) = R_V \left [S(J + Q_k + \Sigma + H) \right ] \ .
\end{equation}


With $H$, we can derive the symmetry constraint on $\Gamma_k$. This identity can be directly derived from the QME, and as such it can be regarded as the translation of the gauge independence of physical observables on the level of the effective average action. We derive it first in the case of Yang-Mills-type theories for simplicity, and we later generalised the result to the case of any gauge theory.

The identity is a main result of this paper, and is summarised in the following theorem.

Before stating and proving the theorem, we need to prove a preliminary Lemma.

\begin{lemma}\label{lemma:anomaly}
        The following relation holds in the renormalised theory, for any local functional $F$ that does not contain antifields and any interaction $\tilde{V}$:
        \begin{multline} \label{eq:renormalized step}
    \{ S(\tilde{V} + F) , I_0 \} 
    - S(F) \cdot_T \{ S(\tilde{V}) , I_0 \} = \frac{i}{\hbar} S(\tilde{V} + F) \cdot_T \left [ \{\tilde{V}, F \}_T - i \hbar A_{\tilde{V}}(F) \right ] \ ,
\end{multline}
where $A_{\tilde{V}}(F) := \An(\tilde{V}+F)- \An(\tilde{V})$.
    \end{lemma}
    \begin{proof}
      We apply the master Ward identity \eqref{eq:MWI}  in the following two cases:

    \begin{equation}\label{eq: case 1}
    \{ S(\tilde{V}) , I_0 \} = \frac{i}{\hbar} S(\tilde{V}) \cdot_{T} \left [ \{\tilde{V} , I_0 \}_{T} +\frac{1}{2}\{\tilde{V},\tilde{V}\}_{T} - i \hbar \An(\tilde{V}) \right ] \ ,
\end{equation}
and
    \begin{multline}\label{eq: case 2}
        \{ S(\tilde{V}+F) , I_0 \} =   \frac{i}{\hbar} S(\tilde{V}+F) \cdot_{T} \Big [ \{\tilde{V}+F , I_0 \}_{T} \\
        +\frac{1}{2}\{\tilde{V}+F,\tilde{V}+F\}_{T} - i \hbar \An(\tilde{V}+F)    \Big ]     
    \end{multline}

Now we use the fact that $F$ does not depend on antifields, so its bracket with $I_0$ and with itself vanishes, and Eq. \eqref{eq: case 2} simplifies to
    \begin{multline}\label{eq: case 2b}
        \{ S(\tilde{V})\cdot_{T}S(F) , I_0 \} =   \frac{i}{\hbar} S(\tilde{V}+F) \cdot_{T} \Big [ \{\tilde{V}, I_0 \}_{T} \\
        +\frac{1}{2}\{\tilde{V},\tilde{V}\}_{T}+\{\tilde{V},F\}_{T} - i \hbar \An(\tilde{V}+F) \Big ]
    \end{multline}
Next, to obtain the r.h.s. of \eqref{eq:renormalized step}, we subtract $S(F)\cdot_T$\eqref{eq: case 1} from \eqref{eq: case 2b} and obtain
\begin{multline}
    \{ S(\tilde{V})\cdot_{T}S(F) , I_0 \} - S(F)\cdot_T\{ S(\tilde{V}) \cdot_{T}, I_0 \}\\
    = \frac{i}{\hbar} S(\tilde{V}+F) \cdot_{T} \Big [\{\tilde{V}, F \}_{T} - i \hbar (  \An(\tilde{V}+F)- \An(\tilde{V}) ) \Big ] \ . 
\end{multline}
 \end{proof}

\begin{theorem}[Effective master equation for the effective average action] \label{theorem:master}
The renormalised Quantum Master Equation \eqref{eq:rQME} implies a symmetry constraint on the effective average action $\Gamma_k(\varphi, \varphi^\ddagger; \phi, \sigma, \eta)$:
\begin{equation} \label{eq:master ren}
\delta_0 \tilde \Gamma_k + \int_x \left [ \frac{\delta \tilde \Gamma_k}{\delta \phi(x)} \frac{\delta \tilde \Gamma_k}{\delta \sigma(x)} + q_k(x) \frac{\delta \tilde \Gamma_k}{\delta \eta(x)} \right ] = i \hbar \mathcal A  \ ,
\end{equation}
where $\tilde \Gamma_k = \Gamma_k + Q_k(\phi)$ and $\mathcal A := \langle A_{\alpha_\Theta(V)} (J + Q_k) \rangle$.
\end{theorem}
\begin{proof}
The proof works by showing that 
\begin{equation} \label{eq:QME-master equation relation}
\frac{i}{\hbar} \langle \mathcal{GRQME} \rangle =  \delta_0 \tilde \Gamma_k + \int_x \left [ \frac{\delta \tilde \Gamma_k}{\delta \phi(x)} \frac{\delta \tilde \Gamma_k}{\delta \sigma(x)} + q_k(x) \frac{\delta \tilde \Gamma_k}{\delta \eta(x)} \right ] - i \hbar  \mathcal A \ ,
\end{equation}
where the rQME with generalised gauge-fixing is $\mathcal{GRQME} = \frac{1}{2} \{ \alpha_\Theta(V), \alpha_\Theta(V) \}_T - i \hbar A(\alpha_\Theta(V) )$.

By definition, we have
\[
0 = \frac{i}{\hbar}\langle \mathcal{GRQME} \rangle = \frac{i}{\hbar} e^{-\frac{i}{\hbar}W_k} S^{-1}(V)\star \left [ S(Q_k+H) \cdot_T S(\alpha_\Theta(V) ) \cdot_T \mathcal{GRQME} \right] \ .
\]
The master Ward identity \eqref{eq:MWI-general} states that $\frac{\hbar}{i} \{S(\alpha_\Theta(V)), I_0 \} = S(\alpha_\Theta(V)) \cdot_T \mathcal{GRQME}$: substituting it to the above relation we get
\[
0 = \frac{i}{\hbar} \langle \mathcal{GRQME} \rangle = e^{-\frac{i}{\hbar}W_k} S^{-1}(V)\star \left [ S(J + Q_k) \cdot_T \{ S(\alpha_\Theta(V)), I_0 \} \right] \ .
\]
Since $F= J+Q_k$ does not contain antifields, we can Lemma \ref{lemma:anomaly} with $\tilde V = \alpha_\Theta(V)$ to arrive at
\begin{multline}\label{eq:RQME-EME-intermediate}
\frac{i}{\hbar} \langle \mathcal{GRQME} \rangle \\
=e^{-\frac{i}{\hbar}W_k} S^{-1}(V)\star \bigg [ \{ S(\alpha_\Theta(V)+ J + Q_k) , I_0 \}  \\
- \frac{i}{\hbar} S(\alpha_\Theta(V)+J + Q_k) \cdot_T \left ( \{\alpha_\Theta(V), J + Q_k \}_T - i \hbar A_{\alpha_\Theta(V)}(J + Q_k) \right ) \bigg ] \ ,
\end{multline}
where we can recognise that $S^{-1}(V)\star \left [ \{ S(\alpha_\Theta(V) + J + Q_k) , I_0 \} \right ] = \delta_0 Z_k$. 

Now, we can individually evaluate the terms proportional to the bracket $\{\alpha_\Theta(V), J + Q_k \}_T$.

The bracket $\{\alpha_\Theta(V), J\}_T$ gives
\[
\{ J , \alpha_\Theta(V) \}_T = \int_x j(x) \frac{\delta}{\delta \varphi^\ddagger(x)} \alpha_\Theta(V) = \int_x j(x) \frac{\delta}{\delta \varphi^\ddagger(x)} V \left ( \varphi^\ddagger + \sigma + \eta \varphi \right ) \ ,
\]
and so
\begin{equation}\label{eq:j-bracket}
\{ J , \alpha_\Theta(V) \}_T = \int_x j(x) \frac{\delta}{\delta \sigma(x)} \alpha_\Theta(V) \ .
\end{equation}

The bracket $\{\alpha_\Theta(V), Q_k\}_T$ in turn is
\[
\{\alpha_\Theta(V) , Q_k \}_T  = - \int_{x} q_k(x) \varphi(x) \frac{\delta}{\delta \varphi^\ddagger(x)} V\left (\varphi^\ddagger + \sigma + \eta \varphi \right ) \ ,
\]
which therefore give
\begin{equation}\label{eq:q-bracket}
    \{\alpha_\Theta(V) , Q_k \}_T  = - \int_x q_k(x) \frac{\delta \alpha_\Theta(V)}{\delta \eta} \ .
\end{equation}

Therefore, substituting Eqs. \eqref{eq:j-bracket} and \eqref{eq:q-bracket} in Eq. \eqref{eq:RQME-EME-intermediate} we arrive at
\begin{equation}
\delta_0 Z_k - \int_x j(x) \frac{\delta Z_k}{\delta \sigma(x)} + q_k(x) \frac{\delta Z_k}{\delta \eta(x)} = \langle A_{\alpha_\Theta(V)} (J + Q_k) \rangle \ .
\end{equation}

From the above expression, the result follows by noticing that 
\[
j(x) = - \frac{\delta \tilde \Gamma_k}{\delta \phi} \ , \quad e^{-\frac{i}{\hbar}W_k} \frac{\delta Z_k}{\delta \sigma} = \frac{i}{\hbar} \frac{\delta W_k}{\delta \sigma} = \frac{i}{\hbar}\frac{\delta \tilde \Gamma_k}{\delta \sigma} 
\ ,
\]
and similarly $e^{-\frac{i}{\hbar}W_k} \frac{\delta Z_k}{\delta \eta} = \frac{\delta \tilde \Gamma_k}{\delta \eta}$, and  $e^{- \frac{i}{\hbar}W_k} \delta_0 Z_k = \frac{i}{\hbar} \delta_0 W_k = \frac{i}{\hbar} \delta_0 \tilde \Gamma_k$.
\end{proof}

\begin{remark}
Due to the modification \eqref{def:extended-Zk}, the scale dependence of the generating functional $Z_k$ is introduced by the extended BV differential $\mBV$, via $H+Q_k = \mBV \Xi$. Even though the field $\eta$ enters the generating functional only through the trivial part of the cohomology, the same does not hold for the regulator function $q_k$ and the scale $k$. This implies that, although physical observables do not depend on $\eta$, they still depend non-trivially on $k$.

A way of seeing this is expressing the flow equations for $\Gamma_k$ as
\[
\partial_k \Gamma_k = \langle \partial_k \mBV \Xi \rangle - \partial_k Q_k(\phi) \ .
\]
Therefore, the $k-$derivative of the effective average action is not a $\mBV-$exact term, precisely because the scale dependence comes through the extended BV differential. This ensures that physical vertex functions derived from $\Gamma_k$ depend non-trivially on $k$.
\end{remark}

\begin{remark} \label{remark:Z-J form}
The effective master equation \eqref{eq:master ren} is written to emphasise the contribution coming from the regulator term. However, we can introduce a source for the variation of $\eta$, by including a term $\Sigma_\eta = \int_x \sigma_\eta q_k$ in $Z_k$, where, by definition, $\frac{\delta \tilde \Gamma_k}{\delta \sigma_\eta} = q_k$. Doing so allows us to rewrite the effective master equation in Zinn-Justin form, with the addition of the Koszul-Tate differential, which in the absence of anomalies read
\begin{equation}
\delta_0 \tilde \Gamma_k + \int_x \frac{\delta \tilde \Gamma_k}{\delta \phi(x)} \frac{\delta \tilde \Gamma_k}{\delta \sigma(x)} = 0 \ .
\end{equation}

The source for the variation of $\eta$ introduces an additional $k-$dependence in the generating functional. There are two ways to deal with this additional term in the flow equation: i) one can simply evaluate the flow for $\sigma_\eta = 0$; or ii) since $\sigma_\eta q_k$ is a classical contribution, we can subtract it from the effective average action via a redefinition $\Gamma_k \to \Gamma_k - \int_x \sigma_\eta q_k $. This term removes the additional $k-$dependence in the Wetterich equation and we get the same flow equation \eqref{eq:wetterich}.

\end{remark}

\subsubsection{Anomalies}\label{sec:anomalies}

Equation \eqref{eq:master ren} tells us that, in the absence of anomalies, $\tilde \Gamma_k$ is invariant under the symmetry transformation
\begin{equation}
\delta_\theta \phi = \frac{\delta \tilde \Gamma_k}{\delta \sigma} \theta \ , \ \delta_\theta \eta = q_k \theta \ .
\end{equation}
In general this differs from $s\phi=\frac{\delta V(\phi)}{\delta \varphi^\ddagger}$. To see this, we note that $\frac{\delta \tilde \Gamma_k}{\delta \sigma}=\frac{\delta \Gamma_k}{\delta \sigma}$, since $ \tilde \Gamma_k$ and $ \Gamma_k$ differ by a term independent of $\sigma$. Next we compute
$$\frac{\delta \Gamma_k}{\delta \sigma} = \left\langle \frac{\delta V}{\delta \varphi^\ddagger} \right\rangle\,,$$
and observe that in general
\begin{equation}
\left\langle \frac{\delta V}{\delta \varphi^\ddagger} \right\rangle \neq \frac{\delta V(\phi)}{\delta \varphi^\ddagger} \,,
\end{equation}
but these coincide if $V$ is linear both in fields and antifields. This agrees with the known result saying that if an action is invariant under linear symmetries, then the (average) effective action is also invariant under these same symmetries. In fact, the linear symmetry $s \eta = q_k \ , s q_k = 0$ is inherited by $\tilde \Gamma_k$.

In the presence of anomalies, the classical symmetry $s_k (I + J + \Sigma + Q_k + H) =0$ does not hold in the renormalised theory (that is, for the effective average action).

The anomaly term $\mathcal A(V)$, thus, measures the breaking of the extended BV invariance at the quantum level. This anomalous contributions can arise due to the Epstein-Glaser renormalization of the time-ordered products, which may not be compatible with the extended symmetry of the classical action. This is analogous to the concept of anomaly discussed in \cite{Brunetti2022}.

The anomaly term in the effective master equation arises from $A_{\alpha_\Theta(V)}(J + Q_k)$, which in the case of Yang-Mills-type theories becomes $A_{\alpha_\Theta(V)}(J + Q_k) = A(V + J + \Sigma + Q_k + H) - A(V+ \Sigma + H)$. Now, the argument of the first contribution to the anomaly is $s_k-$closed by Eq. \eqref{eq:classical extended BV invariance}. Moreover, the regulator and sources terms are $s_k-$exact, since they are given in Eq. \eqref{eq:source-regulator-exact}, and so they belong to the trivial cohomology of $s_k$. Similarly, the second contribution $A(V + \Sigma + H) = A(\alpha_\Theta(V))$ depends on the interaction $V$ and a generalised gauge-fixing sector $\Sigma + H$. Standard arguments show that the gauge-fixing sectors are trivial in cohomology \cite{Barnich2000}. The anomaly is therefore constrained by the minimal sector of the cohomology of $s_k$.

Since the pairs $(j, \sigma)$ and $(q_k,\eta)$ are contractible, it follows that the cohomology of $s_k$ coincides with the cohomology of the BV differential $s$. Therefore, anomalies are controlled by standard BV cohomology.

In some situations, one can remove the anomaly by inductive redefinition of time-ordered products \cite{Brennecke2008}, if certain cohomological conditions are met. These conditions are fulfilled for Yang-Mills in the absence of chiral fermions and in 4D pure gravity, see e.g. \cite{Barnich2000, Fredenhagen2013, Hollands2008,Rejzner2015}.

\section{Solution of the effective master equation}
Since the effective master equation \eqref{eq:master ren} has the same algebraic structure of the Zinn-Justin equation, one can make use of standard perturbative methods to constrain the form of the renormalized action and prove perturbative renormalizability by the same methods known in the literature, see e.g. \cite{Zinn-Justin2002}. In this section, we are interested in how the symmetry constrains the structural form of $\tilde \Gamma_k$, i.e. its functional dependence on the fields and the sources. The general form of the effective average action compatible with the symmetry will then be used as the input to solve the flow equation, i.e. to find the trajectories of the coupling constants in the parameter space under rescalings of $k$.

Following the discussion in Ref. \cite{Barnich2000}, here we propose a non-perturbative method to solve the effective master equation \eqref{theorem:master}. As in the case of flow equations, even though the derivation was done for off-shell functionals, i.e., at the algebraic level, in order to find explicit solutions of the effective master equation we now consider $\Gamma_k$ evaluated on a faithful state $\omega$, i.e., on a particular gauge-fixed field configuration $(\varphi = 0, \varphi^\ddagger = 0)$. This simplifies the discussion, since now $\omega(\delta_0 \Gamma_k) = 0$, and $\omega(\Gamma_k)$ is a functional of external fields $(\phi, \sigma, \eta)$. Moreover, we consider only theories linear in the antifields, as Yang-Mills and gravity.

For simplicity of notation, from now on we will denote the state evaluated effective average action by the same symbol.

The starting point is the observation that, since $\Gamma_k$ is defined as a quantum M\o ller operator, it is a formal power series in $\hbar$ \cite{Duetsch2001,Hawkins2020}. In the classical limit $\hbar \to 0$, from \eqref{eq:classical limit Gamma} it follows that $\tilde \Gamma_k$ reduces to the classical action, with additional terms coming from $Q_k$ and $H$:
\begin{equation} \label{def:extended action}
\tilde \Gamma_k \xrightarrow{\hbar \to 0} I_{ext} = I_0(\phi) + V(\phi, \sigma) + H(\phi, \eta) + Q_k(\phi) \ .
\end{equation}
The dependence of $V$ on the sources $\sigma$ comes from the $\Sigma$ term, which is just a copy of the antifield dependence of $V$.

The above observation suggests to consider the decomposition
\begin{equation} \label{def decomposition}
\tilde \Gamma_k = I_{ext} + \hbar \hat \Gamma_k  \ .
\end{equation}
Notice that we are not assuming an approximation at first loop order, but simply exploiting the fact the zeroth order of $\tilde \Gamma_k$ corresponds to the classical action to separate the $\hbar-$dependent contribution $\hat \Gamma_k$.

Substituting the decomposition \eqref{def decomposition} in \eqref{eq:master ren}, we get
\begin{equation} \label{eq:symmetry constraint decomposition}
\Slav I_{ext} + \hbar \Slav \hat \Gamma_k + \hbar^2 \frac{1}{2}(\hat \Gamma_k , \hat \Gamma_k) = i \hbar \mathcal A(V) \ ,
\end{equation}
where we introduced the effective bracket
\begin{equation}
(A,B) = \int_x \frac{\delta A}{\delta \phi(x)} \frac{\delta B}{\delta \sigma(x)} + (-1)^{\abs{A}} \frac{\delta A}{\delta \sigma(x)} \frac{\delta B}{\delta \phi(x)} \ ,
\end{equation}
defined by declaring the sources to be conjugate to the respective effective fields, i.e., $(\phi^A(x), \sigma^B(y)) = \delta^{AB} \delta(x-y)$, where the indices $A, \ B$ run on the field type, and the operator
\begin{equation}
\Slav A = \int_x \frac{\delta I_{ext}}{\delta \phi(x)} \frac{\delta A}{\delta \sigma} + (-1)^{\abs{A}} \frac{\delta I_{ext}}{\delta \sigma(x)} \frac{\delta A}{\delta \phi(x)} + q_k(x) \frac{\delta A}{\delta \eta(x)} \ .
\end{equation}

Since $\tilde \Gamma_k$ is a formal power series in $\hbar$, to solve the effective master equation, each term in the above decomposition must vanish independently.

The above decomposition suggests to extend the classical BV algebra of the functionals of the original fields and antifields $(\varphi, \varphi^\ddagger)$ to the space of functionals of $(\varphi, \varphi^\ddagger, \phi, \sigma)$ in a natural way, that is, considering the effective fields and BRST sources $(\phi, \sigma)$ as the effective part of the algebra. In this way we complete the transition from the original field configurations $(\varphi, \varphi^\ddagger)$ to the space of effective fields and sources $(\phi, \sigma)$, where a natural notion of BV algebra is inherited from the original structure.

The three conditions coming from \eqref{eq:symmetry constraint decomposition} can now be interpreted as cohomological constraints:
\[
\Slav I_{ext}=0\,,\quad  \Slav \hat \Gamma_k= i \hbar \mathcal A(V) \,,\quad  \frac{1}{2}(\hat \Gamma_k , \hat \Gamma_k)=0\,.
\]
We will now discuss them separately.
\subsection{Effective BV invariance of $I_{ext} $}
We start with
\begin{equation} \label{eq:BRST invariance of I ext}
\Slav I_{ext} = 0 \ .
\end{equation}
This equation is identically satisfied by $I_{ext}$ given by \eqref{def:extended action}, and it encodes the effective BV invariance of $I_{ext}$, where the effective BV transformations are
\begin{equation}
\Slav \phi = \frac{\delta V}{\delta \sigma} \ , \ \Slav \eta = q_k\ , \ \Slav q_k=0\ ,
\end{equation}
Note that the last two equations imply that $q_k$ and $\eta$ form a contractible pair.
This demonstrates that $I_{ext}$ plays the role of the proper solution to the Classical Master Equation in the space of functionals of $(\phi, \sigma,q_k,\eta)$.

\subsection{Cohomology condition} 
\label{subsec:cohomology condition}
Next we consider the linear order, for the moment in the absence of anomalies:
\begin{equation}  \label{eq:cohomological condition}
\Slav \hat \Gamma_k = 0 \ .
\end{equation}
This gives a non-trivial condition on the effective average action. 
As usual, $\Slav$-exact solutions can be re-absorbed by redefinitions of the fields \cite{Anselmi1994, Barnich2000}, and so the non trivial part of $\hat \Gamma_k$ must be in the cohomology of $\Slav$. 

The operator can be decomposed into
\begin{equation}
\Slav = \mBRST + \int_x q_k(x) \frac{\delta}{\delta \eta(x)} \ .
\end{equation} 
Note that the second term acts only on the non-minimal sector and guarantees that $q_k$ and $\eta$ form a contractible pair \cite{Barnich2000}. Hence the non-trivial information about the cohomology of the effective BV operator is already encoded in $\mBRST$, which acts as:
\[
\mBRST \phi = \Slav \phi\,,\quad \mBRST \eta = 0\,.
\]
This is analogous to the action of the BV operator $s$,
but acting on $(\phi, \sigma)$ instead of $(\varphi, \varphi^\ddagger)$. Hence, in particular, the effective antighost and the Nakanishi-Lautrup fields $( \phi_{\bar c} \ , \phi_b)$ form a contractible pair, since $\bar{c}$ and $b$ form a contractible pair.

The quantum contribution $\hat \Gamma_k$ is then in the cohomology of the effective BV operator $\mBRST$.

The cohomology of $\mBRST$ is then determined by its minimal, regulator independent sector.  By standard arguments, the cohomology of $\mBRST$ on local functionals (that is, integrals of local top forms) is characterised by the cohomology group of $\mBRST$ modulo the exterior differential, $H^{g,n}(\mBRST | d)$, where $g$ is the effective-ghost number, the degree associated to the effective field corresponding to ghost fields, and $n$ is the form degree. Here we are working on the space of local $n$-forms, rather than local functionals, as is standard in the literature \cite{Barnich2000}.

Since the action of $\mBRST$ in the space $(\phi, \sigma)$ is identical to the action of $s$ in the space $(\varphi, \varphi^\ddagger)$, the cohomology group $H^{g,n}(\mBRST | d)$ can be characterised by the standard treatment of local BV cohomology \cite{Barnich2000}.

Since $I_{ext}$ has effective-ghost number $0$, in absence of anomalies $\hat \Gamma_k$ is determined by the BRST cohomology in ghost number $0$, $H^{0,n}(\mBRST)$ in the space of functionals of $(\phi, \sigma)$.

In principle, the computation of the BV cohomology provides the most general solution to the effective master equation.
The effective average action is, in general, non-local, so one needs to study the cohomology of the BRST operator, not restricted to local functionals as in \cite{Barnich2000}. However, one of the most used non-perturbative truncations for the effective average action is the \textit{derivative expansion},
\begin{equation} \label{eq:DE}
\Gamma_k =  \int_x f(\phi, \partial_{\mu} \phi,..., \partial_{(\mu_{i}}... \partial_{\mu_{n})} \phi)
\end{equation}
up to some finite order $n$ in the order of the derivative. At each order in the truncation, the effective average action is local, and one can apply the theorems on local BV cohomology; only including derivative of infinite order one is able to explicitly consider non-local effects. The assumption within this truncation is that non-local effects can be parameterised by $k-$dependent coefficients in front of the derivative expansion. Restricting the effective average action to the functional form \eqref{eq:DE}, implies that we can solve the effective master equation by standard local BV cohomology techniques.

In the space of local functionals, the local BRST cohomology $H^{0,n}(\mBRST | d)$ is completely determined by powerful theorems \cite{Barnich1994, Barnich1994-2, Barnich2000, Barnich1994-3}.

\begin{remark} \label{remark: cohomological treatment of anomalies}
In the presence of an anomaly term $\mathcal A(V)$, the identity \eqref{eq:BRST invariance of I ext} guarantees that the anomaly is at least of order $\hbar$:
\[
i   \hbar \mathcal A(V) = \sum_{n \geq 1} \hbar^n a_n(V) \ ,
\]
and the cohomology condition \eqref{eq:cohomological condition} gets modified into
\[
\mathscr S \hat \Gamma_k = a_1 \ .
\]
General theorems guarantee that the lowest order anomaly term in $\hbar$ is local \cite{Barnich2000}. The above condition provides the first-order condition on $ \mathcal A(V)$, since it implies
\[
\Slav a_1 = 0 \ .
\]
Trivial anomalies in the form $a_1 = \Slav \Theta$ can be reabsorbed in local counter-terms in the action. Non-trivial anomalies therefore are constrained by the cohomology of top forms in ghost number 1, $H^{1,d}(s|d)$. This relevant cohomology class has been studied extensively in the literature; e.g., in the case of 4D pure gravity or Yang-Mills theories, it is known \cite{Barnich1994-WZ} that $H^1(s|d)$ is trivial. At least for Yang-Mills-type theories, then, the anomaly can be removed. This is of course consistent with the discussion in section \ref{sec:anomalies}.
\end{remark}

\subsection{Consistency condition} \label{subsec:consistency}
\begin{equation} \label{eq:consistency condition}
(\hat \Gamma_k , \hat \Gamma_k) = 0 \ .
\end{equation}
This condition tells us that, treating the regulator term as a generalised gauge-fixing, the correction $\hat \Gamma_k$ itself satisfies the same Zinn-Justin equation as the effective action without regulator.

This condition can be discussed once the cohomological constraint \eqref{eq:cohomological condition} is solved, as an identity or as an additional, non-trivial constraint. For example, in effective-ghost number $0$, general theorems for Yang-Mills-type theories \cite{Barnich2000} guarantee that all representatives of $H^{0,d}(\gamma | d)$, where $d$ the space-time dimension, can be chosen to be strictly gauge-invariant, except for Chern-Simons forms for $d$ odd, and in particular the BRST sources $\sigma$ can be removed in all counter-terms. In this case, the condition \eqref{eq:consistency condition} reduces to a trivial identity, since $\frac{\delta \hat \Gamma_k}{\delta \sigma} = 0$. Similar considerations apply also for $H^{1,d}(\gamma | d)$, which controls gauge anomalies and where \eqref{eq:consistency condition} can provide a non-trivial constraint on the anomaly term.

It is clear that, once the structural form of $\tilde \Gamma_k$ is determined through the symmetry constraint, $\Gamma_k$ is determined by a simple translation.

\section{Consistency of the symmetries with the RG flow and flow of composite operators} \label{sec:consistency}
Theorem~\ref{theorem:master} gives us powerful cohomological methods to constrain the structural form of the effective average action. However, equation \eqref{eq:master ren} and its most important consequence, \eqref{eq:cohomological condition}, are useful only if they are compatible with the RG flow equation \eqref{eq:wetterich}, i.e., if solving the equation at fixed scale $k_0$ implies that the condition is satisfied at all scales.

To prove that the effective master equation is preserved along the RG flow, we consider the slightly more general problem, of the flow for a composite, local operator $\mathcal O_k(x; \varphi, \partial \varphi,...)$, that is, an operator with arbitrary dependence on the field $\varphi$ and its derivatives, which might also depend on the cut-off parameter $k$. A strategy \cite{Pagani2016} to compute its flow equation is to further extend the definition of the generating functional $Z_k$, to the generating functional of classical composite operators coupled with an external source $\upsilon$:
\begin{equation}
Z_k(\varphi, \varphi^\ddagger ; j, \sigma, \eta, \upsilon) = S(V)^{-1} \star \left [ S \left (V+ Q_k + J + \Sigma + H + \Upsilon_k \right ) \right ] \ ,
\end{equation}
where $\Upsilon_k = \int_x \mathcal O_k(x) \upsilon(x)$. The definition for the $\upsilon-$dependent $W_k(\upsilon, j)$ carries on as in section \ref{sec:generating functionals}.
To keep the notation reasonably compact, from now on we will suppress the dependence on the fields and antifields and on the sources $(\sigma, \eta)$. Note that $Z_k(j, \sigma, \eta) = Z_k(j, \sigma, \eta, \upsilon = 0)$. We will then show that
\begin{proposition}
Given a composite operator $\mathcal O_k(x; \varphi, \partial \varphi,...)$, its flow equation is given by
\begin{equation} \label{eq:flow-composite-operators}
\partial_k \langle \mathcal O_k \rangle = - \frac{i\hbar}{2} \int_x 	\partial_k q_k(x) G_k(x,x) \frac{\delta^2 \langle \mathcal O_k \rangle}{\delta \phi(x) \delta \phi(x)} G_k(x,x) + \langle \partial_k \mathcal O_k \rangle \ .
\end{equation}
\end{proposition}
\begin{proof}

The derivation of the RG flow equation for $Z_k(j, \upsilon)$ follows the same steps as in section \ref{sec:flow-eqs}, with an additional term coming from the $k-$dependent $\mathcal O_k$:
\begin{equation}
\partial_k Z_k(j, \upsilon) = \frac{i}{\hbar} S(V)^{-1} \star \left [ S \left (V+Q_k+J + \Sigma + H + \int_x \mathcal O_k \upsilon \right ) \cdot_T \left ( \partial_k Q_k + \partial_k \Upsilon_k \right ) \right ] \ ,
\end{equation}
and analogously
\begin{equation}
\partial_k W_k(j, \upsilon) = \langle \partial_k Q_k \rangle_{\upsilon} + \langle \partial_k \Upsilon_k \rangle_\upsilon \ .
\end{equation}

The flow equation for the $\upsilon-$dependent effective action gets modified accordingly:
\begin{equation} \label{eq:flow-eq-epsilon}
\partial_k \Gamma_k = \frac{i\hbar}{2} \int_x \partial_k q_k G_k(\phi, \upsilon) + \langle \partial_k \Upsilon_k \rangle_ \upsilon \ .
\end{equation}
To get the flow equation for the operator $\mathcal O_k$, we notice that
\begin{equation}
\frac{\delta W_k}{\delta \upsilon} = \langle \mathcal O_k \rangle_{\upsilon} \ .
\end{equation}
Since the derivatives of $W_k$ and $\Gamma_k$ with respect to $\upsilon$ coincide, we have
\begin{equation}
\frac{\delta \Gamma_k}{\delta \upsilon} = \langle \mathcal O_k \rangle_{\upsilon}
\end{equation}
where the mean value operator is taken for $j = j_\phi$.

Therefore, taking the derivative with respect to $\upsilon$ of the RG flow equation \eqref{eq:flow-eq-epsilon} gives the flow for the expectation value of the composite operator $\mathcal O_k$:
\begin{multline}
\partial_k \frac{\delta \Gamma_k(\upsilon)}{\delta \upsilon} = - \frac{i\hbar}{2} \int_x 	\partial_k q_k(x) G_k(x,x) \frac{\delta \Gamma_k^{(2)}}{\delta \upsilon} G_k(x,x) \\
+ \frac{i}{\hbar} \left ( \langle \mathcal O_k \cdot_T \partial_k \Upsilon \rangle_{\upsilon} - \langle \mathcal O_k \rangle_\upsilon \langle \partial_k \Upsilon \rangle_\upsilon \right ) + \langle \partial_k \mathcal O_k \rangle_{\upsilon} \ .
\end{multline}
Evaluating for vanishing source $\upsilon = 0$ we arrive at the proposition.
\end{proof}

We now return to the problem of compatibility of the effective master equation with the flow. We can apply the flow equation for composite operators \eqref{eq:flow-composite-operators} to $\langle \mathcal{RQME} \rangle$, which is equivalent to the operator in the l.h.s of the effective master equation thanks to \eqref{eq:QME-master equation relation}.
The last term in the flow of composite operators \eqref{eq:flow-composite-operators} vanishes, since $\mathcal{RQME}$ does not exhibit explicit $k-$dependence. Then, we have
\begin{equation}
\partial_k \langle \mathcal{RQME} \rangle = - \frac{i \hbar}{2}  \int_x 	\partial_k q_k(x) G_k(x,x) \frac{\delta^2 \langle \mathcal{RQME} \rangle }{\delta \phi(x) \delta \phi(x)} G_k(x,x) \ .
\end{equation}
Since the flow of the operator is proportional to itself, if the effective master equation is satisfied at some scale $k$ $\langle \mathcal{RQME} \rangle = 0$ it is automatically satisfied at all scales.

\section{Conclusions and Outlook}

In this paper we introduced a Wetterich-type flow equation for gauge theories quantised using the BV formalism. The language of pAQFT allows to treat in a unified way theories in generic, globally hyperbolic spacetimes, without referring to a particular state. With respect to our previous paper \cite{DDPR2022}, we expanded the treatment of fRG in pAQFT to gauge theories, discussing the formalism at the algebraic level only.

The extension of the field configuration space to include the regulator function $q_k$ and a new auxiliary field $\eta$ allows to derive extended Slavnov-Taylor identities for the effective average action, in which the regulator sector appears as an extension of the gauge-fixing term. This in turn allows to treat the regulator sector as a non-minimal contribution in cohomology, and to constrain the effective average action by the powerful techniques of BV and BRST cohomology \cite{Barnich2000}.

In this construction, the locality of the regulator function $q_k$ is crucial, in order to interpret it as an additional field. 


The solution of the cohomological problem at some fixed RG scale $k$ can then be used as an input for the flow equation.
We have also shown the consistency of the extended Slavnov-Taylor identities with the RG flow: once they are satisfied at a certain scale $k = \Lambda$, they are satisfied at all scales.



In future works, it will be of interest to investigate a genuinely non-perturbative formulation of the effective average action, derived in the context of dynamical $C^*-$algebras, and the connections with the non-perturbative St\"uckelberg-Petermann group \cite{Brunetti2021UMWI, Brunetti2021, Buchholz2019}. Along the same lines, it would be interesting to relate the effective action arising from the Epstein-Glaser renormalization of time-ordered products, and the effective action as a solution of the RG flow equation.




The general formalism we developed here can be now be applied to concrete examples, to study the non-perturbative renormalizability of gauge theories in curved, Lorentzian spacetimes. 

The main difficulty in concrete applications is the construction of the interacting propagator $G_k$, just as in the Euclidean case. 

In the Lorentzian case the construction depends on the choice of a state \cite{DDPR2022}. A state for the free theory is, thus, the last missing ingredient to apply our formalism to concrete models and derive the $\beta-$functions \cite{DDPR2022}. In the formalism of the present paper, the choice of a (quasifree) state for the free (linearised) theory corresponds to the choice of a 2-point function $\Delta_+$, together with the evaluation of the functional on the vanishing field configuration $\varphi = 0$. 

Once a state is chosen, it is possible to investigate the RG flow of some gauge theory on curved spacetime.
The study of a free state for Yang-Mills theory in curved spacetimes is a delicate issue, beyond the scope of this paper. Based on recent constructions of states for linearised Yang-Mills theories and gravity in curved spacetimes (see e.g. \cite{ Gerard2023, Gerard2014,Gerard2024,  Murro2024, Wrochna2014}), in future works we will investigate the fRG flow of Yang-Mills theories in curved spacetimes. 

The viability of this formalism in concrete applications has been tested to investigate the existence of a non-trivial fixed point and the asymptotic safety scenario in the RG flow of Lorentzian quantum gravity \cite{DAngelo2023}.


\section*{Data statement}
Data sharing is not applicable to this article as no new data were created or analysed in this study.

\section*{Acknowledgements}
We are grateful to Nicolò Drago, Markus Fr\"ob, Nicola Pinamonti, and Berend Visser for useful discussions on the paper and for a thorough reading of the manuscript and to Antonio D. Pereira for helpful remarks. We are also grateful for the kind hospitality of Perimeter Institute during the completion of the manuscript.
K.R. would like to thank Astrid Eichhorn, Benjamin Knorr, Alessia Platania, Frank Saueressig and Marc Schiffer for very helpful and inspiring discussions.
E.D. is supported by a PhD scholarship of the University of Genoa, by the project GNFM-INdAM Progetto Giovani \textit{Non-linear sigma models and the Lorentzian Wetterich equation}, CUP\textunderscore E53C22001930001, and is the recipient of a INdAM scholarship to conduct research abroad.
\printbibliography
\end{document}